%% file: main.tex
\documentclass[11pt]{article}

\usepackage{amsmath}

\usepackage{amsfonts}
\usepackage{amsthm}
\usepackage{amssymb}
\usepackage[utf8]{inputenc}

\usepackage{textcomp}

\usepackage[margin=1in]{geometry}
\usepackage{indentfirst}

\usepackage{authblk}
\usepackage{abstract}

\usepackage{pbox}
\usepackage{comment}
\usepackage{adjustbox}
\usepackage[linesnumbered, boxruled]{algorithm2e}
\usepackage{pgfplots}
\pgfplotsset{compat=1.4}
\usetikzlibrary{shapes}
\usetikzlibrary{plotmarks}

\usepackage[noadjust]{cite}
\usepackage{framed}
\usepackage{mdframed}
\usepackage{placeins}
\usepackage{subcaption}
\usepackage{afterpage}
\usepackage[hidelinks]{hyperref}
\usepackage{courier}
\usepackage{multirow}
\usepackage{mathrsfs}
\usepackage{microtype} 
\usepackage{dsfont}

\usetikzlibrary{patterns}
\usetikzlibrary{arrows.meta, patterns.meta}

\allowdisplaybreaks

\newcommand{\eps}{\varepsilon}

\newcommand{\eff}{\mathcal{F}}

\newcommand{\rr}{\mathbb{R}}
\newcommand{\qq}{\mathbb{Q}}
\newcommand{\rrp}{\rr_{\ge 0}}

\newcommand{\vpp}{\mathcal{P}_2[V]}

\DeclareMathOperator{\supp}{supp}

\DeclareMathOperator{\dist}{dist}

\newtheorem{theorem}{Theorem}
\newtheorem{lemma}[theorem]{Lemma}
\newtheorem{fact}[theorem]{Fact}
\newtheorem{corollary}[theorem]{Corollary}
\newtheorem{definition}[theorem]{Definition}

\definecolor{purple}{RGB}{51,34,136}
\definecolor{blue}{RGB}{136,204,238}
\definecolor{teal}{RGB}{68,170,153} 
\definecolor{green}{RGB}{17,119,51}
\definecolor{gold}{RGB}{153,153,51} 
\definecolor{yellow}{RGB}{221,204,119}
\definecolor{orange}{RGB}{204,102,119} 
\definecolor{red}{RGB}{136,34,85}
\definecolor{pink}{RGB}{170,68,153}

\numberwithin{theorem}{section}

\usepackage{amsmath}

\newtheorem*{definition*}{Definition}
\newtheorem*{theorem*}{Theorem}

\newtheorem*{lemma*}{Lemma}

\title{On the Structure of Unique Shortest Paths in Graphs}
\author{Greg Bodwin\thanks{bodwin@umich.edu.  Work performed in part while employed by MIT EECS.} }
\affil{University of Michigan EECS}
\date{}

\begin{document}
\maketitle
\thispagestyle{empty}

\input{abstract}

\pagebreak

\thispagestyle{empty}
\tableofcontents
\setcounter{page}{0}
\pagebreak

\input{intro}

\input{metrize}

\input{acks}

\FloatBarrier
\bibliographystyle{plain}
\bibliography{references}

\end{document}

%% file: abstract.tex
\begin{abstract}
We study the combinatorial structure of systems of unique shortest paths in real-weighted graphs.
We say that such a path system is \emph{strongly metrizable}.

A folklore fact is that every strongly metrizable path system must be \emph{simple} and \emph{consistent}, meaning that it avoids paths that repeat nodes, and it also avoids the pattern where a pair of its paths intersect, split apart, and then intersect again later.
Our contribution is to fully characterize strong metrizability via an expanded list of forbidden subsystems, beyond the two implied by simplicity and consistency.
That is:
\begin{itemize}
\item We define a new family of patterns that we call \emph{polyhedral pairs}, which are derived from 2-colored abstract polyhedra,

\item We prove that a directed path system is strongly metrizable via a directed graph if and only if it is simple, consistent, and it does not contain the nontrivial image of either side of an \emph{oriented} polyhedral pair as a subsystem,

\item We prove that an undirected path system is strongly metrizable via an undirected graph if and only if it is simple, consistent, and it does not contain the nontrivial image of either side of a \emph{non-oriented} polyhedral pair as a subsystem.
\end{itemize}
We also discuss some aesthetic structural properties that can be forced for these polyhedra, and we discuss improvements to the characterization that can be obtained in the directed acyclic setting.
\end{abstract}

%% file: intro.tex

\section{Introduction}

Many problems in graph theory and algorithms seek to understand the shortest paths of an input graph.
In the setting where the input graph can have arbitrary real edge weights, it is typical to assume without loss of generality that the graph has a \emph{unique} shortest path between all pairs of nodes, e.g.\ because we may add a tiny random variable to each edge weight to break ties.
The goal of this paper is to study the combinatorial structure of these unique shortest path systems.
Formally, we are interested in which \emph{path systems} are \emph{strongly metrizable}:

\begin{definition} [Path Systems]
A \emph{path system} is a pair $S = (V, \Pi)$, where $V$ is a set whose elements are called ``nodes'' and $\Pi$ is a set of nonempty sequences of nodes called ``paths.''
\end{definition}

\begin{definition} [Strong Metrizability] \label{def:introsm}
A path system $S = (V, \Pi)$ is strongly metrizable if there exists a directed graph $G = (V, E, w)$, with arbitrary real edge weights,\footnote{If $G$ contains a directed cycle of nonpositive weight, then none of these nodes can participate in a unique shortest path, since a path can always take another lap around this cycle without increasing its length.  Thus, while not strictly necessary in our formalisms, one may essentially require w.l.o.g.\ that all directed cycles in $G$ have positive weight.} such that each path $\pi \in \Pi$ is the unique shortest path between its endpoints in $G$.
\end{definition}

For (non-)examples, it is well known folklore that neither of the following two small path systems is strongly metrizable:

\begin{figure}[h!]
\begin{center}
\begin{tikzpicture}

\begin{scope}[yshift=-0.55cm]
  \draw[fill=black] (-4,0) circle[radius=0.15cm];
  \node at (-4.4,0) {$a$};

  \draw[ultra thick, ->, blue]
    (-4,0.08) arc[start angle=-90, end angle=270, radius=0.55cm];
\end{scope}

\node [align=center] at (-4,-1.5) {the self-loop:\\$V = \{a\}, \Pi = \{(aa)\}$};

\draw[fill=black] (0,0) circle[radius=0.15cm];
\node at (-0.2,0.40) {$a$};

\draw[fill=black] (4.3,0) circle[radius=0.15cm];
\node at (4.5,0.40) {$c$};

\draw[ultra thick, ->, red]
  (-0.5,0.2) to[bend right=30]
  node[
    midway,
    circle,
    fill=black,
    inner sep=0pt,
    minimum size=0.3cm,
    label={[text=black]above:{$b$}}
  ] {}
  (5,0.4);

\draw[ultra thick, ->, blue]
  (-0.5,-0.2) to[bend left=30] (5,-0.4);

\node [align=center] at (2.25,-1.5) {the inconsistency:\\$V=\{a,b,c\}, \Pi=\{(ac),(abc)\}$};
\end{tikzpicture}
\end{center}
\caption{\label{fig:consistent} Neither ``the self-loop'' nor ``the inconsistency'' is strongly metrizable.}
\end{figure}

These small path systems may each be viewed as a gesture towards a much broader class of path systems that ``contain'' the depicted system, and thus are not strongly metrizable.
Let us say that a path system $S$ is \emph{simple} if none of its paths contain the same node more than once, and $S$ is \emph{consistent} if, for any pair of distinct nodes $(a, c)$, all of its paths use the same $a \leadsto c$ subpath.
It is well-known folklore that every strongly metrizable path system must be simple and consistent.
The self-loop and inconsistency are the minimal examples of path systems that fail to be simple and consistent, respectively.
We will formalize this direction using the following notion of \emph{subsystems}:
\begin{definition} [Subsystems]
A path system $S = (V, \Pi)$ is a \emph{subsystem} of another path system $S' = (V', \Pi')$, written $S \subseteq S'$, if one can reach $S$ from $S'$ by zero or more applications of the following operations:
\begin{itemize}
\item Delete a path $\pi \in \Pi'$,

\item Delete a node $v \in V$ (and delete $v$ from all paths that contain it; for example, a path $(uvxy)$ would become $(uxy)$),

\item Choose a path $\pi \in \Pi'$ and an internal vertex $v \in \pi$ (neither the first nor the last vertex of $\pi$).
Split $\pi$ into two new paths: a contiguous prefix ending at $v$, and a contiguous suffix beginning at $v$.
\end{itemize}
\end{definition}

A fact implicit in prior work is that strong metrizability is \emph{inherited} by subsystems, in the sense of the following theorem.
\begin{theorem} [Folklore]
If $S$ is strongly metrizable and $S' \subseteq S$, then $S'$ is also strongly metrizable.
\end{theorem}
\begin{proof}
Let $G = (V, E, w)$ be a graph witnessing strong metrizability of $S$.
When we delete or split a path in $S$, it is clear that the paths of $S$ are still unique shortest paths in $G$.
When we delete a vertex $v$, we modify $G$ with a \emph{weighted contraction} of $v$: for each current in-neighbor $u$ and out-neighbor $x$ of $v$, add a new edge $(u, x)$, with weight set to $w(u, x) := w(u, v) + w(v, x)$ (or, if the edge $(u, x)$ is already present, set its weight to the minimum of its current value and this new value).
After this modification, the paths of $S$ will still be unique shortest paths in $G$.
\end{proof}

Since strong metrizability is hereditary under subsystems, we may therefore hope for a \emph{forbidden subsystem} characterization of strong metrizability.
In other words, our aim is to fully list the minimal path systems that are consistent but not strongly metrizable.
Thus, it will follow that a general path system $S$ is strongly metrizable if and only if it is consistent and it avoids all subsystems from our list.

\subsection{New Forbidden Subsystems}

The starting point for our list is the following path system.
It is simple and consistent, but not strongly metrizable, meaning that it reflects a new forbidden subsystem for unique shortest paths.

\boldmath
\begin{figure}[h!]
\begin{center}
\begin{tikzpicture}
\draw [fill=black] (0, 0) circle [radius=0.15cm];
\draw [fill=black] (1, 0) circle [radius=0.15cm];
\draw [fill=black] (0, 1) circle [radius=0.15cm];
\draw [fill=black] (1, 1) circle [radius=0.15cm];
\draw [fill=black] (0, 2) circle [radius=0.15cm];
\draw [fill=black] (1, 2) circle [radius=0.15cm];

\node [left=0.15cm] at (0, 0) {$a$};
\node [left=0.15cm] at (0, 1) {$c$};
\node [left=0.15cm] at (0, 2) {$e$};
\node [right=0.15cm] at (1, 0) {$b$};
\node [right=0.15cm] at (1, 1) {$d$};
\node [right=0.15cm] at (1, 2) {$f$};

\draw [ultra thick, red, ->] plot [smooth] coordinates {(0, -0.5) (0, 2.5)};
\draw [ultra thick, gold, ->] plot [smooth] coordinates {(-0.5, -0.5) (0, 0) (1, 1) (1, 2) (1, 2.5)};
\draw [ultra thick, green, ->] plot [smooth] coordinates {(1, -0.5) (1, 0) (1, 1) (0, 2) (-0.5, 2.5)};
\draw [ultra thick, blue, ->] plot [smooth] coordinates {(1.5, -0.5) (1, 0) (0, 1) (1, 2) (1.5, 2.5)};

\node at (0.5, -1) {$\pi_r = (ace), \pi_y = (adf), \pi_g = (bde), \pi_b = (bcf)$};
\end{tikzpicture}
\end{center}
\caption{\label{fig:tangle} A new consistent path system that is not strongly metrizable.}
\end{figure}
\unboldmath

\begin{theorem} \label{thm:notangle}
The path system $S$ depicted\footnote{Acknowledgment: this and other figures in this paper use a color-safe palette of Tol \cite{Tol21}.} in Figure \ref{fig:tangle} is not strongly metrizable.
\end{theorem}
\begin{proof}
Suppose towards contradiction that there is a directed graph $G = (V, E, w)$ in which all paths in $S$ are unique shortest paths.
We then have the following four inequalities, obtained by comparing the lengths of our four paths to alternate paths among the same endpoints:
\begin{enumerate}
\item [$\pi_r$:] $w(a,c) + w(c,e) < w(a,d) + w(d,e)$
\item [$\pi_y$:] $w(a,d) + w(d,f) < w(a,c) + w(c,f)$
\item [$\pi_g$:] $w(b,d) + w(d,e) < w(b,c) + w(c,e)$
\item [$\pi_b$:] $w(b,c) + w(c,f) < w(b,d) + w(d,f)$
\end{enumerate}
However, each term appears exactly once on the left and once on the right across these four inequalities.
Thus, summing the inequalities and canceling terms leads to the contradiction $0<0$.
\end{proof}

In order to find additional forbidden subsystems, it is helpful to reimagine the previous proof by the picture in Figure \ref{fig:tangleoct}.
The path system $S$ corresponds to the $2$-colored octahedron shown in this figure in a sense that we will now make precise.

\FloatBarrier

\boldmath
\begin{figure}[h]
\begin{center}
\begin{tikzpicture}[scale=1.25]
\draw [fill=black] (0, 0) circle [radius=0.1cm];
\draw [fill=black] (1, 0) circle [radius=0.1cm];
\draw [fill=black] (0, 1) circle [radius=0.1cm];
\draw [fill=black] (1, 1) circle [radius=0.1cm];
\draw [fill=black] (0, 2) circle [radius=0.1cm];
\draw [fill=black] (1, 2) circle [radius=0.1cm];

\node [left=0.15cm] at (0, 0) {$a$};
\node [left=0.15cm] at (0, 1) {$c$};
\node [left=0.15cm] at (0, 2) {$e$};
\node [right=0.15cm] at (1, 0) {$b$};
\node [right=0.15cm] at (1, 1) {$d$};
\node [right=0.15cm] at (1, 2) {$f$};

\draw [ultra thick, red, ->] plot [smooth] coordinates {(0, -0.5) (0, 2.5)};
\draw [ultra thick, gold, ->] plot [smooth] coordinates {(-0.5, -0.5) (0, 0) (1, 1) (1, 2) (1, 2.5)};
\draw [ultra thick, green, ->] plot [smooth] coordinates {(1, -0.5) (1, 0) (1, 1) (0, 2) (-0.5, 2.5)};
\draw [ultra thick, blue, ->] plot [smooth] coordinates {(1.5, -0.5) (1, 0) (0, 1) (1, 2) (1.5, 2.5)};

\draw [fill=black] (2.5, 0.8) -- (2.5, 1.2) -- (4.5, 1.2) -- (4.5, 1.5) -- (4.8, 1) -- (4.5, 0.5) -- (4.5, 0.8) -- cycle;
\end{tikzpicture}%
\hspace{1cm}%
\begin{tikzpicture}[scale=0.75]

\draw [fill=gold, opacity=0.7] (-2, 0) -- (0.4, 0.4) -- (0, 2) -- cycle;
\draw [fill=blue, opacity=0.7] (0.4, 0.4) -- (2, 0) -- (0, -2) -- cycle;
\draw [fill=gray, opacity=0.7] (0.4, 0.4) -- (2, 0) -- (0, 2) -- cycle;
\draw [fill=gray, opacity=0.7] (0.4, 0.4) -- (-2, 0) -- (0, -2) -- cycle;

\draw [fill=red, opacity=0.7] (-0.4, -0.4) -- (0, 2) -- (2, 0) -- cycle;
\draw [fill=green, opacity=0.7] (-0.4, -0.4) -- (-2, 0) -- (0, -2) -- cycle;
\draw [pattern=north west lines, pattern color=black, thick] (-2, 0) -- (-0.4, -0.4) -- (0, 2) -- cycle;
\draw [pattern=north east lines, pattern color=black, thick] (-0.4, -0.4) -- (2, 0) -- (0, -2) -- cycle;

\draw [fill=black] (-2, 0) circle [radius=0.15cm];
\draw [fill=black] (2, 0) circle [radius=0.15cm];
\draw [fill=black] (0, 2) circle [radius=0.15cm];
\draw [fill=black] (0, -2) circle [radius=0.15cm];
\draw [fill=black] (-0.4, -0.4) circle [radius=0.15cm];
\draw [fill=black] (0.4, 0.4) circle [radius=0.15cm];

\draw [ultra thick] (-2, 0) -- (0, 2) -- (2, 0) -- (0, -2) -- (-2, 0);
\draw [ultra thick] (-2, 0) -- (-0.4, -0.4) -- (2, 0);
\draw [ultra thick] (0, 2) -- (-0.4, -0.4) -- (0, -2);

\draw [ultra thick, dashed] (-2, 0) -- (0.4, 0.4) -- (2, 0);
\draw [ultra thick, dashed] (0, 2) -- (0.4, 0.4) -- (0, -2);

\node at (0, 2.4) {$a$};
\node at (-0.65, -0.65) {$e$};
\node at (0, -2.4) {$b$};
\node at (-2.4, 0) {$d$};
\node at (2.4, 0) {$c$};
\node at (0.65, 0.65) {$f$};

\end{tikzpicture}
\end{center}
\caption{\label{fig:tangleoct} A path system $S$ that is not strongly metrizable, derived from the 2-coloring of the octahedron shown here.  The color of each colorful face matches that of its corresponding path in $S$.}
\end{figure}
\unboldmath

\begin{definition} [Abstract Polyhedron]
We say that an \emph{abstract polyhedron} $Q$ is a compact connected nonempty $2$-manifold without boundary with a cell decomposition of its surface.
We say that $Q$ is (non)-orientable if the manifold is (non)-orientable.

A two-coloring of an abstract polyhedron $Q$ is an assignment of one of two colors to each of its faces such that no two faces that share a nontrivial boundary arc have the same color.
We will use the convention that the two colors are ``gray'' and ``colorful;'' this lets us use colors to nuance the colorful faces, as we did in Figure \ref{fig:tangleoct}.
\end{definition}

\begin{definition} [Oriented Polyhedral Pairs] \label{def:polyhedral}
Let $T = (V, \Pi), T' = (V, \Pi')$ be distinct path systems over the same vertex set $V$.
We say that $(T, T')$ form an \emph{oriented polyhedral pair} if there exists an orientable $2$-colored abstract polyhedron $Q$ with vertices $V$, such that:
\begin{enumerate}
\item The nodes in each path $\pi \in \Pi$ are exactly the nodes on some corresponding colorful cell $q \in Q$.
Moreover, the order of the nodes of $\pi$ agrees with the orientation of $q$ (e.g., they occur around $q$ in clockwise order).

\item The nodes in each path $\pi' \in \Pi'$ are exactly the nodes of some corresponding gray cell $q' \in Q$.
Moreover, the order of the nodes of $\pi'$ agrees with the reversed orientation of $q'$ (e.g., they occur around $q$ in counterclockwise order). 

\item (Arc Agreement Property) Let us say that an arc $(u, v)$ on a cell $c$ is an ``endpoint arc'' if it corresponds to the endpoints of the path $\pi$ associated to $c$, rather than a pair of nodes $(u, v)$ that occur consecutively on $\pi$.
The last property is that the set of endpoint arcs induces a perfect matching between the gray and colorful faces; in particular, each cell has exactly one endpoint arc on its boundary.
\end{enumerate}
\end{definition}

To clarify these definitions, let us revisit Figure \ref{fig:tangleoct}.
The colorful path system $T$ on the left may be viewed as half of a polyhedral pair; the paired system $T'$ (not pictured) would be formed by taking the alternate three-node paths among the four endpoint pairs used by $T$ (which use the opposite node in the middle layer $\{c, d\}$).
These satisfy the embedding properties into the two-colored octahedron pictured on the right: for example, one of the paths in $T$ is $\pi_r = (ace)$, and these are exactly the three nodes in the red face of the octahedron read in clockwise order.
Additionally, we can verify that the arc agreement property is satisfied: for example, $(a, e)$ is an endpoint arc due to $\pi_r$; this matches the red face to the gray face $(ade)$, which does not have any other incident endpoint arcs and so is not matched to any other colorful face.

It is a slight abuse of terminology to call these ``polyhedral'' pairs, since traditional polyhedra are similar but different from cell decompositions, but we will revisit and justify this choice further in a moment.

\begin{theorem} \label{thm:intropolysmone}
For any polyhedral pair $(T, T')$, if a path system $S$ contains either $T$ or $T'$ as a subsystem, then $S$ is not strongly metrizable.
\end{theorem}

One application of Theorem \ref{thm:intropolysmone} is that it lets us easily generate new forbidden subsystems for strong metrizability, by considering various two-colored polyhedra.
Further examples are shown in Figures \ref{fig:hbp}, \ref{fig:esp}, and \ref{fig:tor}.

\begingroup
\boldmath

\def\DiagramSideWidth{5.9cm}
\def\DiagramGap{4.5mm}
\def\DiagramRowWidth{15cm}

\def\AlignedCentralBlackArrow{%
  \begin{tikzpicture}[baseline=(current bounding box.center)]
    \path[fill=black]
      (-1.15,-0.20) -- (-1.15,0.20) --
      ( 0.85, 0.20) -- ( 0.85,0.50) --
      ( 1.15, 0.00) --
      ( 0.85,-0.50) -- ( 0.85,-0.20) -- cycle;
  \end{tikzpicture}%
}

\providecolor{gold}{RGB}{190,145,0}
\definecolor{pathgold}{RGB}{190,145,0}
\tikzset{
  path arrow/.style={
    ultra thick,
    ->,
    >=stealth,
    shorten <=3.1pt,
    shorten >=3.1pt
  }
}

\begin{figure}[!p]
\centering
\setlength{\abovecaptionskip}{2pt}
\setlength{\belowcaptionskip}{0pt}
\begin{adjustbox}{max totalsize={\textwidth}{0.96\textheight},center}
\begin{minipage}{\textwidth}
\centering
\noindent\makebox[\textwidth][c]{%
\vbox{%
  \offinterlineskip
  \hbox{%
    \makebox[\DiagramSideWidth][r]{%
      \begin{tikzpicture}[baseline=(current bounding box.center)]
\draw [fill=black] (0, 0) circle [radius=0.15cm];
\draw [fill=black] (1,0) circle [radius=0.15cm];

\draw [fill=black] (-1, 2) circle [radius=0.15cm];
\draw [fill=black] (0.5, 2) circle [radius=0.15cm];
\draw [fill=black] (2, 2) circle [radius=0.15cm];

\draw [fill=black] (-1, -2) circle [radius=0.15cm];
\draw [fill=black] (0.5, -2) circle [radius=0.15cm];
\draw [fill=black] (2, -2) circle [radius=0.15cm];

\draw [ultra thick, red, ->] plot [smooth] coordinates {(-1, -2) (0, 0) (-1, 2)};
\draw [ultra thick, orange, ->] plot [smooth] coordinates {(0.5, -2) (1, 0) (-1, 2)};
\draw [ultra thick, gold, ->] plot [smooth] coordinates {(0.5, -2) (0, 0) (0.5, 2)};
\draw [ultra thick, green, ->] plot [smooth] coordinates {(2, -2) (1, 0) (0.5, 2)};
\draw [ultra thick, blue, ->] plot [smooth] coordinates {(2, -2) (0, 0) (2, 2)};
\draw [ultra thick, purple, ->] plot [smooth] coordinates {(-1, -2) (1, 0) (2, 2)};

\node at (-1, -2.5) {$a$};
\node at (0.5, -2.5) {$b$};
\node at (2, -2.5) {$c$};

\node [left=0.1cm] at (-0.1, 0) {$d$};
\node [right=0.1cm] at (1.1, 0) {$e$};

\node at (-1, 2.5) {$f$};
\node at (0.5, 2.5) {$g$};
\node at (2, 2.5) {$h$};
      \end{tikzpicture}%
    }%
    \hskip\DiagramGap
    \AlignedCentralBlackArrow
    \hskip\DiagramGap
    \makebox[\DiagramSideWidth][l]{%
      \begin{tikzpicture}[baseline=(current bounding box.center)]
\draw [fill=purple, opacity=0.7] (1.2, 0.4) -- (0, -2.5) -- (0.25, 1) -- cycle;
\draw [fill=gray, opacity=0.7] (1.2, .4) -- (0.25, 1) -- (0, 2.5) -- cycle;
\draw [fill=gray, opacity=0.7] (1.2, .4) -- (1, -0.6) -- (0, -2.5) -- cycle;
\draw [fill=red, opacity=0.7] (0, 2.5) -- (0.25, 1) -- (-1, 0.6) -- cycle;
\draw [fill=gray, opacity=0.7] (0.25, 1) -- (0, -2.5) -- (-1, 0.6) -- cycle;
\draw [fill=orange, opacity=0.7] (-1.2, -0.4) -- (-1, 0.6) -- (0, -2.5) -- cycle;

\draw [fill=blue, opacity=0.7] (0, 2.5) -- (1.2, 0.4) -- (1, -0.6) -- cycle;
\draw [fill=green, opacity=0.7] (1, -0.6) -- (0, -2.5) -- (-0.25, -1) -- cycle;
\draw [fill=gold, opacity=0.7] (0, 2.5) -- (-0.25, -1) -- (-1.2, -0.4) -- cycle;
\draw [pattern=vertical lines, thick] (0, 2.5) -- (-1, 0.6) -- (-1.2, -0.4) -- cycle;
\draw [pattern=north east lines, thick] (0, 2.5) -- (-0.25, -1) -- (1, -0.6) -- cycle;
\draw [pattern=north west lines, thick] (0, -2.5) -- (-1.2, -0.4) -- (-0.25, -1) -- cycle;

\draw [fill=black] (0, 2.5) circle [radius=0.15cm];
\draw [fill=black] (0, -2.5) circle [radius=0.15cm];
\draw [fill=black] (0.25, 1) circle [radius=0.15cm];
\draw [fill=black] (-0.25, -1) circle [radius=0.15cm];
\draw [fill=black] (-1, 0.6) circle [radius=0.15cm];
\draw [fill=black] (1.2, 0.4) circle [radius=0.15cm];
\draw [fill=black] (1, -0.6) circle [radius=0.15cm];
\draw [fill=black] (-1.2, -0.4) circle [radius=0.15cm];

\draw [thick, dashed] (0, 2.5) -- (0.25, 1);
\draw [thick] (0, 2.5) -- (-0.25, -1);
\draw [thick] (0, 2.5) -- (-1, 0.6);
\draw [thick] (0, 2.5) -- (1, -0.6);
\draw [thick] (0, 2.5) -- (1.2, 0.4);
\draw [thick] (0, 2.5) -- (-1.2, -0.4);

\draw [thick, dashed] (0, -2.5) -- (0.25, 1);
\draw [thick] (0, -2.5) -- (-0.25, -1);
\draw [thick, dashed] (0, -2.5) -- (-1, 0.6);
\draw [thick] (0, -2.5) -- (1, -0.6);
\draw [thick, dashed] (0, -2.5) -- (1.2, 0.4);
\draw [thick] (0, -2.5) -- (-1.2, -0.4);

\draw [thick] (1.2, 0.4) -- (1, -0.6) --  (-0.25, -1) -- (-1.2, -0.4) -- (-1, 0.6);
\draw [thick, dashed] (-1, 0.6) -- (0.25, 1) -- (1.2, 0.4);

\node at (0, 2.9) {$d$};
\node at (1.5, 0.7) {$h$};
\node at  (1.2, -0.9) {$c$};
\node at (-1.5, -0.7) {$b$};
\node at (0, -2.9) {$e$};
\node at (0.05, 0.6) {$a$};
\node at (-0.45, -0.6) {$g$};
\node at (-1.2, 0.9) {$f$};
      \end{tikzpicture}%
    }%
  }%
  \vskip6mm
  \hbox to \DiagramRowWidth{\hfil
    $\pi_r = (adf), \pi_o = (bef), \pi_y = (bdg),
     \pi_g = (ceg), \pi_b = (cdh), \pi_p = (aeh)$%
  \hfil}%
}%
}

\caption{A new forbidden subsystem, derived from the hexagonal bipyramid.}
\label{fig:hbp}

\vspace{2pt}
\noindent\makebox[\textwidth][c]{%
\vbox{%
  \offinterlineskip
  \hbox{%
    \makebox[\DiagramSideWidth][r]{%
      \begin{tikzpicture}[scale=0.7,baseline=(current bounding box.center)]
\draw [fill=black] (0, 0) circle [radius=0.15cm];
\draw [fill=black] (0, 1) circle [radius=0.15cm];

\draw [fill=black] (3, 0) circle [radius=0.15cm];
\draw [fill=black] (3, 1) circle [radius=0.15cm];

\draw [fill=black] (1.5, -1) circle [radius=0.15cm];

\draw [fill=black] (1.5, 2) circle [radius=0.15cm];
\draw [fill=black] (0, 3) circle [radius=0.15cm];
\draw [fill=black] (0, 4) circle [radius=0.15cm];
\draw [fill=black] (3, 3) circle [radius=0.15cm];
\draw [fill=black] (3, 4) circle [radius=0.15cm];

\draw [ultra thick, red, ->] (0, -0.5) -- (0, 4.5);
\draw [ultra thick, orange, ->] (3, -0.5) -- (3, 4.5);
\draw [ultra thick, gold, ->] plot [smooth] coordinates {(1.5, -1) (0, 0) (3, 1)};
\draw [ultra thick, green, ->] plot [smooth] coordinates {(1.5, -1) (3, 0) (0, 1)};
\draw [ultra thick, blue, ->] plot [smooth] coordinates {(1.5, 2) (0, 3) (3, 4)};
\draw [ultra thick, purple, ->] plot [smooth] coordinates {(1.5, 2) (3, 3) (0, 4)};

\node at (1.5, -1.4) {$a$};
\node at (-0.4, 0) {$b$};
\node at (3.4, 0) {$c$};
\node at (-0.4, 1) {$d$};
\node at (3.4, 1) {$e$};
\node at (1.5, 1.6) {$f$};
\node at (-0.4, 3) {$g$};
\node at (3.4, 3) {$h$};
\node at (-0.4, 4) {$i$};
\node at (3.4, 4) {$j$};
      \end{tikzpicture}%
    }%
    \hskip\DiagramGap
    \AlignedCentralBlackArrow
    \hskip\DiagramGap
    \makebox[\DiagramSideWidth][l]{%
      \begin{tikzpicture}[scale=0.56,baseline=(current bounding box.center)]
\draw [fill=gray, opacity=0.7] (-0.7, 2.5) -- (-0.7, -0.4) -- (3.1, -1.2) -- (3.1, 1.7) -- cycle;
\draw [fill=orange, opacity=0.7] (-0.7, 2.5) -- (-0.7, -0.4) -- (-3.1, -1.7) -- (-3.1, 1.2) -- cycle;
\draw [fill=gray, opacity=0.7] (-0.7, -0.4) -- (0, -4) -- (-3.1, -1.7) -- cycle;
\draw [fill=gray, opacity=0.7] (-3.1, 1.2) -- (-0.7, 2.5) -- (0, 4) -- cycle;
\draw [fill=gold, opacity=0.7] (0, -4) -- (-0.7, -0.4) -- (3.1, -1.2) -- cycle;
\draw [fill=purple, opacity=0.7] (0, 4) -- (-0.7, 2.5) -- (3.1, 1.7) -- cycle;

\draw [pattern=north east lines, thick] (-3.1, -1.7) -- (-3.1, 1.2) -- (0.7, 0.5) -- (0.7, -2.4) -- cycle;
\draw [fill=red, opacity=0.7] (0.7, 0.5) -- (0.7, -2.4) -- (3.1, -1.2) -- (3.1, 1.7) -- cycle;
\draw [pattern=north west lines, thick] (0.7, -2.4) -- (3.1, -1.2) -- (0, -4) -- cycle;
\draw [pattern=north west lines, thick] (0.7, 0.5) -- (0, 4) -- (3.1, 1.7) -- cycle;
\draw [fill=green, opacity=0.7] (0, -4) -- (0.7, -2.4) -- (-3.1, -1.7) -- cycle;
\draw [fill=blue, opacity=0.7] (0, 4) -- (0.7, 0.5) -- (-3.1, 1.2) -- cycle;

\node at (0, -4.4) {$a$};
\node at (3.5, -1.2) {$b$};
\node at (-3.5, -1.7) {$c$};
\node at (1, -1.9) {$d$};
\node at (-0.4, -0.8) {$e$};
\node at (0, 4.4) {$f$};
\node at (1, 0.2) {$g$};
\node at (-3.5, 1.2) {$j$};
\node at (3.5, 1.7) {$i$};
\node at (-0.4, 2.1) {$h$};

\draw [fill=black] (0, 4) circle [radius=0.15cm];
\draw [fill=black] (0, -4) circle [radius=0.15cm];

\draw [fill=black] (0.7, 0.5) circle [radius=0.15cm];
\draw [fill=black] (0.7, -2.4) circle [radius=0.15cm];

\draw [fill=black] (3.1, 1.7) circle [radius=0.15cm];
\draw [fill=black] (3.1, -1.2) circle [radius=0.15cm];

\draw [fill=black] (-3.1, 1.2) circle [radius=0.15cm];
\draw [fill=black] (-3.1, -1.7) circle [radius=0.15cm];

\draw [fill=black] (-0.7, 2.5) circle [radius=0.15cm];
\draw [fill=black] (-0.7, -0.4) circle [radius=0.15cm];

\draw [thick] (0, 4) -- (3.1, 1.7) -- (0.7, 0.5) -- (0, 4) -- (-3.1, 1.2) -- (0.7, 0.5) -- (0.7, -2.4) -- (-3.1, -1.7) -- (0, -4) -- (0.7, -2.4) -- (3.1, -1.2) -- (0, -4);
\draw [thick] (-3.1, 1.2) -- (-3.1, -1.7);
\draw [thick] (3.1, 1.7) -- (3.1, -1.2);

\draw [thick, dashed] (0, 4) -- (-0.7, 2.5) -- (-3.1, 1.2);
\draw [thick, dashed] (3.1, 1.7) -- (-0.7, 2.5) -- (-0.7, -0.4) -- (3.1, -1.2);
\draw [thick, dashed] (0, -4) -- (-0.7, -0.4) -- (-3.1, -1.7);
      \end{tikzpicture}%
    }%
  }%
  \vskip6mm
  \hbox to \DiagramRowWidth{\hfil
    $\pi_r = (bdgi), \pi_o = (cehj), \pi_y = (abe),
     \pi_g = (acd), \pi_b = (fgj), \pi_p = (fhi)$%
  \hfil}%
}%
}

\caption{A new forbidden subsystem, derived from the elongated square bipyramid.}
\label{fig:esp}

\vspace{2pt}
\noindent\makebox[\textwidth][c]{%
\vbox{%
  \offinterlineskip
  \hbox{%
    \makebox[\DiagramSideWidth][r]{%
      \begin{tikzpicture}[
        line cap=round,
        line join=round,
        scale=0.96,
        baseline=(current bounding box.center)
      ]
\coordinate (a) at (0,0);
\coordinate (b) at (2.4,3.5);
\coordinate (c) at (4.8,0);
\coordinate (d) at (2.4,1.35);


\draw[path arrow,red]
  (a)
  .. controls (1.35,-0.55) and (3.45,-0.55) .. (c)
  .. controls (4.15,0.55) and (3.25,1.10) .. (d);

\draw[path arrow,pathgold]
  (b)
  .. controls (2.02,2.90) and (2.02,2.00) .. (d)
  .. controls (1.65,1.12) and (0.72,0.56) .. (a);

\draw[path arrow,green!65!black]
  (c)
  .. controls (3.45,0.55) and (1.35,0.55) .. (a)
  .. controls (0.20,1.30) and (1.08,2.80) .. (b);

\draw[path arrow,blue]
  (d)
  .. controls (2.78,2.00) and (2.78,2.90) .. (b)
  .. controls (3.72,2.80) and (4.60,1.30) .. (c);

\foreach \p in {a,b,c,d}{
  \draw[fill=black] (\p) circle[radius=0.11cm];
}

\node[left=0.15cm]  at (a) {$a$};
\node[above=0.15cm] at (b) {$b$};
\node[right=0.15cm] at (c) {$c$};
\node[below=0.16cm] at (d) {$d$};
      \end{tikzpicture}%
    }%
    \hskip\DiagramGap
    \AlignedCentralBlackArrow
    \hskip\DiagramGap
    \makebox[\DiagramSideWidth][l]{%
      \begin{tikzpicture}[
        line cap=round,
        line join=round,
        scale=0.90,
        baseline=(current bounding box.center)
      ]
\coordinate (p00) at (0.0,0.0);
\coordinate (p10) at (1.8,0.0);
\coordinate (p20) at (3.6,0.0);

\coordinate (p01) at (0.9,1.55);
\coordinate (p11) at (2.7,1.55);
\coordinate (p21) at (4.5,1.55);

\coordinate (p02) at (1.8,3.10);
\coordinate (p12) at (3.6,3.10);
\coordinate (p22) at (5.4,3.10);


\draw[fill=green!65!black,opacity=0.7]
  (p00) -- (p10) -- (p01) -- cycle; 

\draw[fill=gray,opacity=0.7]
  (p10) -- (p11) -- (p01) -- cycle; 

\draw[fill=red,opacity=0.7]
  (p10) -- (p20) -- (p11) -- cycle; 

\draw[fill=gray,opacity=0.7]
  (p20) -- (p21) -- (p11) -- cycle; 

\draw[fill=blue,opacity=0.7]
  (p01) -- (p11) -- (p02) -- cycle; 

\draw[fill=gray,opacity=0.7]
  (p11) -- (p12) -- (p02) -- cycle; 

\draw[fill=pathgold,opacity=0.7]
  (p11) -- (p21) -- (p12) -- cycle; 

\draw[fill=gray,opacity=0.7]
  (p21) -- (p22) -- (p12) -- cycle; 

\draw[ultra thick]
  (p00) -- (p20) -- (p22) -- (p02) -- cycle;

\draw[ultra thick] (p00) -- (p10) -- (p20);
\draw[ultra thick] (p01) -- (p11) -- (p21);
\draw[ultra thick] (p02) -- (p12) -- (p22);

\draw[ultra thick] (p00) -- (p01);
\draw[ultra thick] (p10) -- (p01);
\draw[ultra thick] (p10) -- (p11);
\draw[ultra thick] (p20) -- (p11);
\draw[ultra thick] (p20) -- (p21);

\draw[ultra thick] (p01) -- (p02);
\draw[ultra thick] (p11) -- (p02);
\draw[ultra thick] (p11) -- (p12);
\draw[ultra thick] (p21) -- (p12);
\draw[ultra thick] (p21) -- (p22);

\foreach \p in {%
  p00,p10,p20,%
  p01,p11,p21,%
  p02,p12,p22%
}{%
  \draw[fill=black] (\p) circle[radius=0.12cm];
}

\node[below left=0.08cm]  at (p00) {$a$};
\node[below=0.12cm]       at (p10) {$c$};
\node[below right=0.08cm] at (p20) {$a$};

\node[left=0.14cm]        at (p01) {$b$};
\node[below=0.15cm]       at (p11) {$d$};
\node[right=0.14cm]       at (p21) {$b$};

\node[above left=0.08cm]  at (p02) {$c$};
\node[above=0.12cm]       at (p12) {$a$};
\node[above right=0.08cm] at (p22) {$c$};
      \end{tikzpicture}%
    }%
  }%
  \vskip6mm
  \hbox to \DiagramRowWidth{\hfil
    $\pi_r = (acd), \pi_y = (bda), \pi_g = (cab), \pi_b = (dbc)$%
  \hfil}%
}%
}

\caption{\label{fig:tor} A new forbidden subsystem, derived from a two-colored triangulation of the torus. (For visual clarity the nodes $a,b,c$ have been represented in the picture several times; when glued together appropriately, the surface is topologically a torus.)}

\end{minipage}
\end{adjustbox}
\end{figure}
\unboldmath


Now that we have identified the polyhedral pairs as a new class of forbidden subsystems, let us acknowledge that certain \emph{vertex gluing} operations will preserve their status as an obstruction to strong metrizability, even if they technically destroy the polyhedral mapping.
Informally: imagine a polyhedral pair $(T, T')$, derived from (say) a checkerboard pattern imposed on a very large sphere.
Then, choose two nodes $v, v'$ on opposite sides of the sphere, and glue them together.
The resulting path systems $(T, T')$ after this gluing are technically no longer a polyhedral pair.
But, this gluing operation certainly will not restore strong metrizability.
Indeed, from a graphical perspective, the gluing is similar to \emph{adding a constraint} that the edges incident to $v$ and $v'$ use the same edge weights as each other, which can only make it harder to find a witness graph.

Given this, we might have some intuition that in addition to polyhedral pairs being forbidden as subsystems, their \emph{images} (which may map two nodes together) are forbidden as well.
This turns out to \emph{almost} be the case: if one maps too many nodes together, it is possible to trivialize a polyhedral pair $(T, T')$ by mapping $T$ and $T'$ to the same system, which can restore strong metrizability (an example of this is shown in Figure \ref{fig:overglue}).
But so long as we preserve distinctness of $T$ and $T'$ through the mapping, this intuition holds.
We then come to our first main result: with this natural extension to the polyhedral pairs, this \emph{completes} the forbidden-subsystem characterization of strong metrizability in the directed setting.
We state this formally as follows.

\begin{definition} [Path System Mapping]
For a path system $S = (V, \Pi)$ and a map $\phi : V \to V'$, we write $\phi(S) := (V', \phi(\Pi))$, where $\phi(\Pi)$ denotes applying $\phi$ to each path $\pi \in \Pi$ entrywise.
\end{definition}

\begin{theorem} [Main Result, Directed Setting] \label{thm:maindir}
A path system $S$ is strongly metrizable \textbf{if and only if} it is simple, consistent, and for any (oriented) polyhedral pair $(T, T')$ and any vertex mapping $\phi$ with $\phi(T) \ne \phi(T')$, $S$ avoids $\phi(T)$ (and $\phi(T')$, by symmetry) as a subsystem.
\end{theorem}

We remark here that this characterization remains correct even if we restrict attention to polyhedral pairs with certain additional properties that more naturally correspond to polyhedra, such as: every face of the corresponding abstract polyhedron is simple (non-self intersecting) and contains at least three nodes, there are no parallel edges, and $T$ is minimal (in the sense that no proper subsystem of $T$ also participates in a polyhedral pair).
These points are purely cosmetic and a bit nuanced, so we defer their discussion to Section \ref{sec:cleanup}.
The condition $\phi(T) \ne \phi(T')$ is necessary for the correctness of this theorem; a counterexample when $\phi(T) = \phi(T')$ is shown in Figure \ref{fig:overglue}.

Phrased another way, $Q$ is a two-colored abstract polyhedron in the usual sense (with cells corresponding to faces), except that we do not require that $Q$ embeds properly in $\rr^3$ with flat/non-intersecting faces.
We also do not require the surface topology of $Q$ to be a sphere; for example, $Q$ may be a toroidal polyhedron, or have the topology of a Klein bottle, etc.
Despite these differences, we will prefer the view of $Q$ as a polyhedron rather than a manifold, and so we will use terminology like \emph{faces}, \emph{edges}, and \emph{vertices} of $Q$ rather than cells, arcs, and cell intersections.
In Section \ref{sec:cleanup}, we discuss some additional ways that we can impose properties on $Q$ that make it resemble a polyhedron in a traditional sense, e.g., having at least three nodes per face.
These are purely aesthetic additions and are not needed for our main characterization theorems.

\begin{figure}
\begin{center}
\begin{tikzpicture}[scale=1.25]

\draw [ultra thick, red, ->] plot [smooth] coordinates {(0, -0.5) (0, 2.5)};
\draw [ultra thick, yellow!70!orange, ->] plot [smooth] coordinates {(-0.5, -0.5) (0, 0) (1, 1) (1, 2) (1, 2.5)};
\draw [ultra thick, green, ->] plot [smooth] coordinates {(1, -0.5) (1, 0) (1, 1) (0, 2) (-0.5, 2.5)};
\draw [ultra thick, blue, ->] plot [smooth] coordinates {(1.5, -0.5) (1, 0) (0, 1) (1, 2) (1.5, 2.5)};

\draw [fill=black] (0, 0) circle [radius=0.1cm];
\draw [fill=black] (1, 0) circle [radius=0.1cm];
\draw [fill=black] (0, 1) circle [radius=0.1cm];
\draw [fill=black] (1, 1) circle [radius=0.1cm];
\draw [fill=black] (0, 2) circle [radius=0.1cm];
\draw [fill=black] (1, 2) circle [radius=0.1cm];

\node [left=0.15cm] at (0, 0) {$a$};
\node [left=0.15cm] at (0, 1) {$c$};
\node [left=0.15cm] at (0, 2) {$e$};
\node [right=0.15cm] at (1, 0) {$b$};
\node [right=0.15cm] at (1, 1) {$d$};
\node [right=0.15cm] at (1, 2) {$f$};

\draw [fill=black] (2.5, 0.8) -- (2.5, 1.2) -- (4.5, 1.2) -- (4.5, 1.5) -- (4.8, 1) -- (4.5, 0.5) -- (4.5, 0.8) -- cycle;

\begin{scope}[xshift=5.8cm]
    \coordinate (a) at (0, 0);
    \coordinate (b) at (1, 0);
    \coordinate (cd) at (0.5, 1);
    \coordinate (e) at (0, 2);
    \coordinate (f) at (1, 2);

    \draw [ultra thick, red, ->] plot [smooth] coordinates {
        (0, -0.5) (a) (0.2, 0.45) (cd) (0.2, 1.55) (e) (0, 2.5)
    };

    \draw [ultra thick, yellow!70!orange, ->] plot [smooth] coordinates {
        (-0.5, -0.5) (a) (0.42, 0.35) (cd) (0.58, 1.65) (f) (1, 2.5)
    };

    \draw [ultra thick, green, ->] plot [smooth] coordinates {
        (1, -0.5) (b) (0.58, 0.35) (cd) (0.42, 1.65) (e) (-0.5, 2.5)
    };

    \draw [ultra thick, blue, ->] plot [smooth] coordinates {
        (1.5, -0.5) (b) (0.8, 0.45) (cd) (0.8, 1.55) (f) (1.5, 2.5)
    };

    \draw [fill=black] (a) circle [radius=0.1cm];
    \draw [fill=black] (b) circle [radius=0.1cm];
    \draw [fill=black] (cd) circle [radius=0.1cm];
    \draw [fill=black] (e) circle [radius=0.1cm];
    \draw [fill=black] (f) circle [radius=0.1cm];

    \node [left=0.15cm] at (a) {$a$};
    \node [right=0.15cm] at (b) {$b$};
    \node [right=0.18cm] at (cd) {$c=d$};
    \node [left=0.15cm] at (e) {$e$};
    \node [right=0.15cm] at (f) {$f$};
\end{scope}
\end{tikzpicture}
\end{center}
\caption{\label{fig:overglue} An example demonstrating the necessity of the $\phi(T) \ne \phi(T')$ condition in Theorem \ref{thm:maindir}.  The path system $S$ on the left is not strongly metrizable, as previously shown, but a mapping that identifies the middle vertices $c$ and $d$ gives a system that is strongly metrizable (since there are no competing alternate paths between any pair of endpoints), thus destroying its status as a forbidden subsystem.}
\end{figure}

\subsection{The Undirected Setting}

So far, we have discussed the setting where the input path system is directed, and the graph witnessing strong metrizability can be directed.
We will next explain how our characterization adapts to the undirected setting.

\begin{definition} [Undirected Path Systems]
An \emph{undirected path system} $S = (V, \Pi)$ is a directed path system in which every path $\pi$ is considered equivalent to the path with reversed node order.
\end{definition}

We will then seek a forbidden-subsystem characterization of strong metrizability of undirected path systems (which requires the witness graph $G$ to be undirected).
The definition of subsystems extends immediately to the undirected setting.
Note that \emph{the inconsistency} in the undirected setting refers to the same path system described previously, but this is now considered equivalent to the systems obtained by reversing path orders.
For example, the system with paths $\{(ac), (cba)\}$ is strongly metrizable in the directed setting, but equivalent to the inconsistency (and therefore not strongly metrizable) in the undirected setting.

The previous definition of \emph{oriented} polyhedral pairs $(T, T')$ does not extend to the undirected setting: it requires that the order of nodes in each path in $T$ agrees with the orientation of the corresponding face of the polyhedron, and that the paths in $T'$ disagree with the corresponding orientation, and this constraint is no longer meaningful when paths are equivalent to their reverse.
We will thus consider the following relaxation:

\begin{definition} [Non-Oriented Polyhedral Pairs]
Let $T = (V, \Pi), T' = (V, \Pi')$ be undirected path systems over the same vertex set $V$.
We say that $(T, T')$ form a non-oriented polyhedral pair if Definition \ref{def:polyhedral} holds with the following modifications:
\begin{itemize}
\item The abstract polyhedron $Q$ may be orientable or non-orientable,

\item The nodes in each path in $\Pi$ and $\Pi'$ still must occur in order around the corresponding colorful and gray face (respectively) of $Q$, but they may occur in either order around the face, and do not need to respect any orientation (if one exists),

\item The arc agreement property holds with respect to each \emph{unordered} edge $\{u, v\} \in Q$: either the adjacent gray and colorful faces both correspond to paths that contain $\{u, v\}$ consecutively, or they both correspond to paths with endpoints $\{u, v\}$.
\end{itemize}
\end{definition}

We may then characterize undirected strong metrizability as follows:

\begin{theorem} [Main Result, Undirected Setting] \label{thm:mainintround}
An undirected path system $S$ is undirected strongly metrizable \textbf{if and only if} it is simple, consistent, and for any non-oriented polyhedral pair $(T, T')$ and any vertex mapping $\phi$ with $\phi(T) \ne \phi(T')$, $S$ contains neither $\phi(T)$ nor $\phi(T')$ as a subsystem.
\end{theorem}

The upshot of this theorem is that the new forbidden subsystems, which apply in the undirected setting but not the directed setting, are precisely those derived from two-colorings of \emph{non-orientable} surfaces.
One such example is given by Figure \ref{fig:nonorientable}.

\begin{figure}[t]
\begin{center}
\begin{tikzpicture}
\tikzset{
	pathlabel/.style={font=\footnotesize\bfseries, inner sep=1pt, fill=white, fill opacity=0.85, text opacity=1},
	gridlabel/.style={font=\footnotesize\bfseries, text=white, inner sep=0pt}
}

\begin{scope}[shift={(-8, 0)}]

\draw [ultra thick, green] (0.5, 0) -- (0.5, 0.5) -- (0.5, 4) -- (0, 4.5);
\draw [ultra thick, yellow!70!orange] (0, 0) -- (0, 0.5) -- (0, 4) -- (0.5, 4.5);
\draw [ultra thick, red] (0, 0) -- (0.5, 0.5) -- (3, 2) -- (2.5, 2.5) -- (2.5, 2);
\draw [ultra thick, orange] (0.5, 0) -- (0, 0.5) -- (2, 2.5) -- (2, 2);
\draw [ultra thick, blue] (2, 2) -- (2.5, 2.5) -- (0.5, 4) -- (0.5, 4.5);
\draw [ultra thick, purple] (2.5, 2) -- (2, 2.5) -- (0, 4) -- (0, 4.5);

\foreach \x/\y in {
	0/0, 0/0.5, 0.5/0, 0.5/0.5,
	0/4, 0/4.5, 0.5/4, 0.5/4.5,
	2/2, 2/2.5, 2.5/2, 2.5/2.5
}{
	\draw [fill=black] (\x, \y) circle [radius=0.15cm];
}

\node[pathlabel, anchor=north east] at (-0.12, -0.12) {a};
\node[pathlabel, anchor=south east] at (-0.12, 0.62) {b};
\node[pathlabel, anchor=north west] at (0.62, -0.12) {c};
\node[pathlabel, anchor=south west] at (0.62, 0.62) {d};

\node[pathlabel, anchor=north east] at (-0.12, 3.88) {e};
\node[pathlabel, anchor=south east] at (-0.12, 4.62) {f};
\node[pathlabel, anchor=north west] at (0.62, 3.88) {g};
\node[pathlabel, anchor=south west] at (0.62, 4.62) {h};

\node[pathlabel, anchor=north east] at (1.88, 1.88) {i};
\node[pathlabel, anchor=south east] at (1.88, 2.62) {j};
\node[pathlabel, anchor=north west] at (2.62, 1.88) {k};
\node[pathlabel, anchor=south west] at (2.62, 2.62) {l};

\draw [ultra thick, ->] (4, 2.25) -- (6.1, 2.25);

\end{scope}

\draw [fill=gray] (0, 1) -- (0, 4) -- (4, 4) -- (4, 1) -- cycle;

\draw [ultra thick] (0, 1) -- (0, 4) -- (4, 4) -- (4, 1) -- cycle;
\draw [ultra thick] (1, 1) -- (1, 4);
\draw [ultra thick] (2, 1) -- (2, 4);
\draw [ultra thick] (3, 1) -- (3, 4);
\draw [ultra thick] (0, 1) -- (4, 1);
\draw [ultra thick] (0, 2) -- (4, 2);
\draw [ultra thick] (0, 3) -- (4, 3);

\draw [fill=red] (1, 3) -- (1, 4) -- (2, 4) -- (2, 3) -- cycle;
\draw [fill=yellow!70!orange] (0, 2) -- (0, 3) -- (1, 3) -- (1, 2) -- cycle;
\draw [fill=blue] (1, 1) -- (2, 1) -- (2, 2) -- (1, 2) -- cycle;
\draw [fill=green, <-] (3, 3) -- (2, 3) -- (2, 2) -- (3, 2) -- cycle;
\draw [fill=orange] (3, 3) -- (4, 3) -- (4, 4) -- (3, 4) -- cycle;
\draw [fill=purple] (3, 1) -- (4, 1) -- (4, 2) -- (3, 2) -- cycle;

\foreach \x in {0,...,4}{
	\foreach \y in {1,...,4}{
		\draw [fill=black] (\x, \y) circle [radius=0.18cm];
	}
}

\draw [ultra thick] (0.8, 3.5) -- (1.2, 3.5);
\draw [ultra thick] (2.8, 3.5) -- (3.2, 3.5);
\draw [ultra thick] (0.8, 2.5) -- (1.2, 2.5);
\draw [ultra thick] (2.8, 2.5) -- (3.2, 2.5);
\draw [ultra thick] (0.8, 1.5) -- (1.2, 1.5);
\draw [ultra thick] (2.8, 1.5) -- (3.2, 1.5);

\draw [ultra thick, ->] (0, 0.55) -- (4, 0.55);
\draw [ultra thick, <-] (0, 4.45) -- (4, 4.45);
\draw [ultra thick, <-] (-0.35, 1) -- (-0.35, 4);
\draw [ultra thick, <-] (-0.55, 1) -- (-0.55, 4);
\draw [ultra thick, <-] (4.35, 1) -- (4.35, 4);
\draw [ultra thick, <-] (4.55, 1) -- (4.55, 4);

\foreach \x/\y/\name in {
	0/1/j, 0/2/e, 0/3/b, 0/4/j,
	1/1/i, 1/2/h, 1/3/a, 1/4/k,
	2/1/l, 2/2/g, 2/3/d, 2/4/l,
	3/1/k, 3/2/f, 3/3/c, 3/4/i,
	4/1/j, 4/2/e, 4/3/b, 4/4/j
}{
	\node[gridlabel] at (\x, \y) {\name};
}

\node at (6, 2.5) {\includegraphics[scale=0.2]{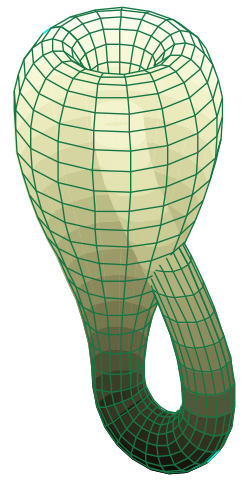}};

\end{tikzpicture}
\end{center}
\caption{\label{fig:nonorientable} The undirected path system on the left is part of a non-orientable polyhedral pair.  This can be seen by mapping its vertices into the two-colored grid shown in the middle, with opposite sides of the grid glued in the orientation indicated by the external arrows (endpoint edges are marked with hashes).  This has the topology of a Klein bottle (Klein bottle figure from \cite{tttrung2006kleinbottle}).}
\end{figure}

\subsection{Structural Corollaries of the Topological Perspective}

The connections between topology and shortest paths developed in this paper have several implications for technology transfer between the areas.
The following two corollaries state facts that were already known in the literature on shortest paths, but our characterization implies new topologically-minded proofs that may be of interest.

\begin{corollary} [Folklore]
If a path system is strongly metrizable, then this fact can always be witnessed by a graph $G = (V, E, w)$ whose edge weights are all positive.
\end{corollary}
\begin{proof} [New Topological Proof Sketch]
The underlying topological fact for this statement is that the boundary of a boundary is always identically $0$.
In particular, by a minor tweak to our main result, the setting with positive edge weights is characterized by two-colored manifolds whose boundary (if nonempty) consists entirely of non-endpoint edges on the boundary of colorful faces.
Since the boundary of this boundary is always $0$, it may always be partitioned into an edge-disjoint union of directed cycles.
We may thus iteratively patch these directed cycles of colorful arcs by adding new gray faces to the surface in a way that respects their orientation.
After these patches are applied, our surface has no remaining boundary and thus corresponds to a two-colored polyhedron exactly as before.
Hence strong metrizability and positive strong metrizability are both equivalently characterized by simplicity, consistency, and forbidden images of polyhedral pairs, and so are equivalent.
\end{proof}

\begin{corollary} [Folklore]
Let $S = (V, \Pi)$ be an undirected path system, and let $\overleftrightarrow{S} = (V, \overleftrightarrow{\Pi})$ be the directed path system that contains each path $\pi \in \Pi$ in both of its possible directions.
Then $S$ is strongly metrizable (by an undirected graph) if and only if $\overleftrightarrow{S}$ is strongly metrizable (by a directed graph).
\end{corollary}
\begin{proof} [New Topological Proof Sketch]
This statement is implied by the topological notion of an \emph{orientable double-cover}.
In particular: for every non-orientable manifold $M_N$, there is an orientable manifold $M_O$ and a $2$-to-$1$ continuous map $f : M_O \to M_N$.
Filtering this fact through Theorems \ref{thm:maindir} and \ref{thm:mainintround}, we learn that for any non-oriented polyhedral pair $(T, T')$, there is always an oriented polyhedral pair $(R, R')$ that maps to $(T, T')$.
Thus strong metrizability of $S$ is characterized by the same set of conditions as strong metrizability of $\overleftrightarrow{S}$, and hence these notions are equivalent.
\end{proof}

One could ask whether the previous corollary can be strengthened: if an undirected system $S$ is \emph{not} strongly metrizable, can we always witness this by choosing some orientation of its paths (choosing one direction per path, as opposed to $\overleftrightarrow{S}$ that takes both directions)?
This turns out to be false: Theorem \ref{thm:mainintround} implies that there are some path systems that display fundamentally undirected obstructions to strong metrizability.

\begin{corollary}
There are undirected path systems $S = (V, \Pi)$ that are not strongly metrizable (by an undirected graph), but where for any directed path system $S' = (V, \Pi')$ obtained by choosing one direction for each path, the system $S'$ is strongly metrizable (by a directed graph).
\end{corollary}
\begin{proof} [Proof Sketch]
Theorems \ref{thm:maindir} and \ref{thm:mainintround} imply that this holds for any non-orientable polyhedral path system $S$ (that does not also contain an orientable polyhedral path system).
An example of such a system is given in Figure \ref{fig:nonorientable}.
\end{proof}

Finally, let us say that a path system is \emph{acyclic} if we can choose a total ordering of its nodes (a \emph{topological ordering}) such that the nodes in each $\pi \in \Pi$ appear in this order.
For the special case of acyclic systems, our characterization implies the following:

\begin{corollary} \label{cor:introdagrot}
Let $S = (V, \Pi)$ be a directed acyclic strongly metrizable path system, and let $S' = (V, \Pi')$ be a directed path system where each $\pi' \in \Pi'$ is obtained by taking circular shifts of the node ordering of some corresponding path $\pi \in \Pi$.
Then $S'$ is strongly metrizable if and only if it is consistent.
\end{corollary}

An example of this corollary is shown in Figure \ref{fig:altoct}.
This figure applies a circular shift to the forbidden subsystem $S$ from Figure \ref{fig:tangle} (which was acyclic), and observes that it remains non-strongly metrizable after a circular shift to its node ordering, since it is generated from the same two-colored octahedron.

\boldmath
\begin{figure}[h]
\begin{center}
\begin{tikzpicture}[scale=0.9]
\draw [fill=black] (0.5, 0) circle [radius=0.15cm];
\draw [fill=black] (0, 1) circle [radius=0.15cm];
\draw [fill=black] (1, 1) circle [radius=0.15cm];
\draw [fill=black] (0, 2) circle [radius=0.15cm];
\draw [fill=black] (1, 2) circle [radius=0.15cm];
\draw [fill=black] (0.5, 3) circle [radius=0.15cm];

\node at (0.5, -0.5) {$b$};
\node [left=0.15cm] at (0, 1) {$c$};
\node [left=0.15cm] at (0, 2) {$e$};
\node at (0.5, 3.5) {$a$};
\node [right=0.15cm] at (1, 1) {$d$};
\node [right=0.15cm] at (1, 2) {$f$};

\draw [ultra thick, red, ->] plot [smooth] coordinates {(0, 0.6) (0, 1) (0, 2) (0.5, 3) (0.9, 3.3)};
\draw [ultra thick, gold, ->] plot [smooth] coordinates {(1, 0.6) (1, 1) (1, 2) (0.5, 3) (0.1, 3.3)};
\draw [ultra thick, green, ->] plot [smooth] coordinates {(0.1, -0.3) (0.5, 0) (1, 1) (0, 2) (-0.4, 2.3)};
\draw [ultra thick, blue, ->] plot [smooth] coordinates {(0.9, -0.3) (0.5, 0) (0, 1) (1, 2) (1.4, 2.3)};

\begin{scope}[shift={(0, 0.5)}]
\draw [fill=black] (2.5, 0.8) -- (2.5, 1.2) -- (4.5, 1.2) -- (4.5, 1.5) -- (4.8, 1) -- (4.5, 0.5) -- (4.5, 0.8) -- cycle;
\end{scope}
\end{tikzpicture}%
\hspace{1cm}%
\begin{tikzpicture}[scale=0.7]
\begin{scope}[shift={(0, 0.5)}]
\draw [fill=gold, opacity=0.7] (-2, 0) -- (0.4, 0.4) -- (0, 2) -- cycle;
\draw [fill=blue, opacity=0.7] (0.4, 0.4) -- (2, 0) -- (0, -2) -- cycle;
\draw [fill=gray, opacity=0.7] (0.4, 0.4) -- (2, 0) -- (0, 2) -- cycle;
\draw [fill=gray, opacity=0.7] (0.4, 0.4) -- (-2, 0) -- (0, -2) -- cycle;

\draw [fill=red, opacity=0.7] (-0.4, -0.4) -- (0, 2) -- (2, 0) -- cycle;
\draw [fill=green, opacity=0.7] (-0.4, -0.4) -- (-2, 0) -- (0, -2) -- cycle;
\draw [pattern=north west lines, pattern color=black, thick] (-2, 0) -- (-0.4, -0.4) -- (0, 2) -- cycle;
\draw [pattern=north east lines, pattern color=black, thick] (-0.4, -0.4) -- (2, 0) -- (0, -2) -- cycle;

\draw [fill=black] (-2, 0) circle [radius=0.15cm];
\draw [fill=black] (2, 0) circle [radius=0.15cm];
\draw [fill=black] (0, 2) circle [radius=0.15cm];
\draw [fill=black] (0, -2) circle [radius=0.15cm];
\draw [fill=black] (-0.4, -0.4) circle [radius=0.15cm];
\draw [fill=black] (0.4, 0.4) circle [radius=0.15cm];

\draw [ultra thick] (-2, 0) -- (0, 2) -- (2, 0) -- (0, -2) -- (-2, 0);
\draw [ultra thick] (-2, 0) -- (-0.4, -0.4) -- (2, 0);
\draw [ultra thick] (0, 2) -- (-0.4, -0.4) -- (0, -2);

\draw [ultra thick, dashed] (-2, 0) -- (0.4, 0.4) -- (2, 0);
\draw [ultra thick, dashed] (0, 2) -- (0.4, 0.4) -- (0, -2);

\node at (0, 2.4) {$a$};
\node at (-0.65, -0.65) {$e$};
\node at (0, -2.4) {$b$};
\node at (-2.4, 0) {$d$};
\node at (2.4, 0) {$c$};
\node at (0.65, 0.65) {$f$};
\end{scope}
\end{tikzpicture}
\end{center}
\caption{\label{fig:altoct} When we take circular shifts of the node orderings of $S$ (Figure \ref{fig:tangle}), we get a new forbidden subsystem, derived from the same 2-colored octahedron in a different way.
}
\end{figure}
\unboldmath

Finally, we mention a quick algorithmic corollary, for the problem of testing whether or not an input path system is strongly metrizable.
It is easy to do so in polynomial time: one can write a straightforward LP, expressing a choice of edge weights that makes every path strictly shorter than its alternatives (this LP generally has exponentially many constraints, but it admits a separation oracle).
We obtain a more concise and efficient LP, which does not require separation oracles.
See Corollary \ref{cor:algorithm} for details.

\subsection{Other Related Work}

\paragraph{Distance Preservers.}

Shortest path structure has been previously studied through the lens of \emph{distance preservers}, in which the goal is to determine the maximum possible number of edges that might be needed in a subgraph that exactly preserves distances among a given set $P$ of $|P|=p$ demand pairs in an $n$-node input graph.
These were introduced by Coppersmith and Elkin \cite{CE06}; some of the work on this problem includes \cite{BCE05, CE06, BV21, Bodwin21, GR17, CDKL17, CGMW18}.

All of the known upper bounds for distance preservers work only by exploiting simplicity and consistency alone, and thus they generalize to bound the number of edges in \emph{any} set of simple consistent paths in a graph (even if those paths cannot be induced as unique shortest paths under any edge weights).
There have been technical lower bounds \cite{BHT23, BV21} demonstrating that the known bounds have reached the limits of simplicity and consistency, and so any further improved upper bounds for distance preservers must exploit new structure.
This paper provides new structure that could, in principle, circumvent these technical lower bounds.
It is therefore an interesting open problem to use it for this purpose.

\paragraph{Homology Theories for Graphs.}

Our work has technical similarity to the path homologies introduced in \cite{GLMY13} (see also \cite{GLMY14, GLMY15, CM18} for followup work).
This prior work has directly influenced many of the design choices made in this paper, especially those related to the \emph{boundary operators} over paths that are used as a central technical tool in our arguments.
However, our definitions and terminology do not always exactly align with those of \cite{GLMY13}, since our goal is to study shortest paths whereas their goal is to devise a coherent graph homology theory.
It would be interesting to unify these frameworks, perhaps relating the shortest path properties of a graph to its appropriate homology groups.
Some progress in this direction may be found in \cite{Wu17}; here, the author relates cycle decompositions of directed graphs (which bear some technical similarity to our characterization theorems) to a certain homology group of those graphs (in a sense related, but not identical, to the homology theory of \cite{GLMY13}).

\paragraph{Approximate and Weight-Restricted Shortest Paths.}

Since we have found that there are forbidden subsystems in unique shortest paths beyond the inconsistency, a natural followup question is to ask whether there are forbidden subsystems beyond the inconsistency in \emph{approximate} shortest paths.
An interesting recent paper by Cizma and Linial \cite{CL26} has essentially answered this question in the undirected setting.
They showed that every $n$-node consistent undirected path system can be induced as $O(n^{1/2})$-approximate shortest paths in an undirected graph, but there are examples on which one cannot improve this approximation factor beyond $n^{1/2 - o(1)}$.
However, the corresponding question for directed graphs remains an interesting open question.

Another problem studied in followup work, perhaps in a similar spirit, is to investigate the extent to which strong metrizability changes if we bound the edge weights of the witnessing graph.
For example, perhaps we require the edge weights to have bounded aspect ratio, or to be bounded integers.
This problem is investigated in \cite{BBW24}, including analogous questions for approximate shortest paths; many open problems remain.

\paragraph{Geodesic Graphs.}

Not every graph is capable of hosting one of the new forbidden subsystems from this paper.
On sufficiently simple graphs $G$, it could be the case that every consistent path system can be induced as (unique) shortest paths under some edge weights.
In work following publication of this paper, a structural theory of such graphs has started to emerge \cite{CL22, CL23, CCL26, CCL25}; among other directions, this work aims to characterize such graphs via forbidden minors.

Besides the focus on metrizability as a property of graphs rather than path systems, there are some further technical differences between our work and this developing theory.
A minor one is that our focus is more on the directed setting, whereas these papers mainly focus on the undirected setting.
More importantly, these papers focus on \emph{all-pairs} path systems, which must contain exactly one path between each pair of nodes; by contrast, we consider \emph{partial} path systems.
This line of work is active; many interesting open problems in this research direction remain.

\subsection{Paper Outline}

In Section \ref{sec:background}, we introduce preliminary definitions and facts about path systems and strong metrizability.
In Section \ref{sec:refutation}, we prove an initial characterization theorem for strong metrizability, stating roughly that a path system $S$ is strongly metrizable if and only if it is simple, consistent, and there is no other distinct ``weighted'' path system $S'$ that shares its topological boundary.
In Section \ref{sec:simplification}, we obtain our main characterization theorems by relating our previous characterization to polyhedral pairs.

%% file: metrize.tex

\section{Paths, Path Systems, and Strong Metrizability \label{sec:background}}

We will begin with a list of introductory definitions and facts about the objects we study.
In an effort to preserve relative brevity, we will frequently omit or sketch proofs in this section, as they typically follow straightforwardly from the definitions.

\subsection{Definitions about Paths and Path Systems}

Recall from the introduction that a \emph{path} $\pi$ is a nonempty sequence of nodes from a ground set $V$.
A path is \emph{simple} if it does not repeat nodes at all.
If $s$ is the first node and $t$ is the last node of a path $\pi$, then we say that $\pi$ is an $s \leadsto t$ path, or that $\pi$ has endpoints $(s, t)$.
A path is a \emph{cycle} if it contains at least two nodes (counting repeats) and its endpoints are equal, i.e., it is an $s \leadsto s$ path.
It is a \emph{simple cycle} if no other nodes are repeated besides its start/endpoint (note that a simple cycle is not a simple path, since it repeats a node).
For a path $\pi$, we write
\begin{align*}
(u, v) \in \pi &\implies u, v \text{ appear adjacently and in that order in } \pi \\
u < v \in \pi &\implies u \ne v \text{ appear in that order (but not necessarily adjacently) in } \pi \\
u \le v \in \pi &\implies u < v \in \pi \text{ or } u=v \in \pi
\end{align*}
This notation may be ambiguous for a non-simple path, but we will clarify in context as needed.
We say that $\pi'$ is a \emph{subpath} of $\pi$, written $\pi' \subseteq \pi$, if $\pi'$ is a subsequence of $\pi$.
This includes the case where $\pi'$ is a \emph{non-contiguous} subsequence of $\pi$; we will clarify in context when $\pi'$ is specifically a contiguous subpath.
If $u \le v \in \pi$, then we write $\pi[u \leadsto v]$ to denote the contiguous subpath of $\pi$ with endpoints $(u, v)$.

Recall from the introduction that a path system $S = (V, \Pi)$ is \emph{strongly metrizable} if there is a directed weighted graph in which the paths in $S$ are each the unique shortest path between their endpoints.
As is well known, we may assume that the witnessing graph does not contain cycles of nonpositive weight:
\begin{fact} \label{fct:nonpos}
If $S = (V, \Pi)$ is strongly metrizable, then there exists a witnessing graph $G = (V, E, w)$ in which every cycle $c$ has $w(c) > 0$.\footnote{If there is a cycle $C$ with nonpositive weight, and a path $\pi$ that intersects the cycle, then $\pi$ can always take an additional lap around $C$ without getting longer.  Hence $\pi$ would not be considered a unique shortest path between its endpoints.}
\end{fact}

We will also discuss \emph{weighted} path systems $S = (V, \Pi, w)$, where the paths in $\Pi$ are equipped with positive path weights (unweighted systems are equivalent to weighted systems with unit path weights).
Confusingly, this means that there are two distinct notions of path weights in this paper: the formal weight of the path in a path system, and occasionally the total weight of the edges along a path through a particular graph.
The former notion will be more common, but we will take care to differentiate these in context.

We note that these definitions (and the following ones) correspond to the \emph{directed} version of the problem discussed in the introduction; discussion of the undirected setting will come later.

\subsection{Definitions about Algebra}

For a set $Z$, we write $\rr^Z$ to denote the vector space of formal linear combinations of the elements of $Z$ with coefficients in $\rr$.
The positive orthant of $\rr^Z$ is denoted by
$$\rrp^Z := \left\{a \in \rr^Z \ \mid \ a \ge 0 \text{ (entrywise)} \right\}.$$
For any element $z \in Z$, we will use $z$ interchangeably with the canonical basis vector in $\rr^Z$, whose entry indexed by $z$ is $1$ and all other entries are $0$.
To denote distinct ordered pairs from a set $V$, we use permutation notation, writing
$$\vpp := \left\{ (u, v) \ \mid \ u \ne v \in V \right\}.$$
We will use two boundary operators:
\begin{definition}[The Boundary Operator $\partial_2$] \label{def:bd3}
The second boundary operator $\partial_2$ is a map defined as follows: for any path $\pi = (v_1, \dots, v_k)$ we define
$$\partial_2(\pi) := \left( \sum \limits_{i=1}^{k-1} \left(v_i, v_{i+1}\right) \right) - \left(v_1, v_k\right) \in \rr^{\vpp},$$
where terms of the form $(v, v)$ (which arise if $\pi$ repeats a node twice in a row, or starts and ends at the same node) are ignored.
This map then extends linearly: letting $p$ be a formal linear combination of paths with coefficients in $\rr$, i.e.,
$$p = \sum_{\pi} \lambda(\pi) \cdot \pi$$
where each $\lambda(\pi) \in \rr$, then we define
$$\partial_2 \left( p \right) := \sum_{\pi} \lambda(\pi) \cdot \partial_2(\pi).$$
\end{definition}

\begin{definition}[The Boundary Operator $\partial_1$]
The first boundary operator $\partial_1$ is a map defined as follows: for any given $(u, v) \in \vpp$, we have
$$\partial_1(u, v) := -u + v \in \rr^V.$$
This map similarly extends linearly over $\rr^{\vpp}$.
\end{definition}
As is standard, we will often suppress the subscripts and refer to both operators $\partial_2, \partial_1$ as $\partial$, since the distinction is always clear from the argument.
The name ``boundary operator'' is inspired by \cite{GLMY13}, which contains a more in-depth discussion of the relationship between these boundary operators and the ones typically used for simplicial complexes.
For now, we will justify this choice of terminology by observing that a boundary of a boundary is identically zero:

\begin{fact} \label{fct:bdbd}
The composition $\partial_1 \circ \partial_2$ is the zero map.
\end{fact}
\begin{proof}
By linearity, it suffices to show that $\partial(\partial(\pi)) = 0$ for any given path $\pi = (v_1, \dots, v_k)$.
We compute:
\begin{align*}
\partial\left( \partial(\pi) \right) &= \partial \left( \left( \sum \limits_{i=1}^{k-1} \left(v_i, v_{i+1}\right) \right) - \left(v_1, v_k\right) \right)\\
&= \left( \sum \limits_{i=1}^{k-1} \partial\left(v_i, v_{i+1}\right) \right) - \partial\left(v_1, v_k\right)\\
&= \left( \sum \limits_{i=1}^{k-1} -v_i + v_{i+1} \right) - \left(-v_1 + v_k\right)\\
&= \left(-v_1 + v_k \right) - \left(-v_1 + v_k\right)\\
&= 0. \tag*{\qedhere}
\end{align*}
\end{proof}

It will be useful to observe the following relationship between boundaries and path lengths.
We denote by $\langle \cdot, \cdot \rangle$ the standard Euclidean inner product, noting that a weight function $w$ may be viewed as a vector in $\rr^{\vpp}$.

\begin{fact} \label{fct:lenbd}
For any $s \leadsto t$ path $\pi$ in a graph $G = (V, E, w)$, we have
$w(\pi) = \left\langle\partial(\pi) + (s, t), w\right\rangle.$
\end{fact}
\begin{proof}
We compute:
\begin{align*}
w(\pi) &= \sum \limits_{i=1}^{k-1} w(v_i, v_{i+1})\\
&= \left\langle \sum \limits_{i=1}^{k-1} \left(v_i, v_{i+1}\right), w \right\rangle\\
&= \left \langle \partial(\pi) + (s, t), w \right\rangle \tag*{\text{by Definition \ref{def:bd3}.} \qedhere}
\end{align*}
\end{proof}

We will also apply boundary operators over path systems: for $S = (V, \Pi, w)$, we have
$$\partial(S) := \sum \limits_{\pi \in \Pi} w(\pi) \cdot \partial(\pi),$$
with unweighted path systems interpreted as having unit path weights.


\subsection{Reducing Path Systems}

There are a few simple ways in which one can change a path system without changing its boundary.
We will say that a path system is \emph{reduced} when none of the following operations can still be applied.

\paragraph{Removing Isolated Nodes.} For $S = (V, \Pi)$, if there are any nodes $v \in V$ that do not participate in any paths in $\Pi$, then we may delete $v$ from $V$ without affecting the strong metrizability of $S$.

\paragraph{Nontriviality.}

Note that $\partial(\pi) := 0$ in the case where $\pi$ has only one or two nodes.
We will say that such a path is \emph{trivial}, and that a path system is \emph{nontrivial} if all paths with $1$ or $2$ nodes have been removed.
Part of the justification for removing such paths is the following fact:

\begin{fact} \label{fct:pathdelete}
Let $S, S'$ be path systems that differ by a path $\pi$ with $1$ or $2$ nodes.
If both systems are simple and consistent, then $S$ is strongly metrizable if and only if $S'$ is strongly metrizable.
\end{fact}
\begin{proof} [Proof Sketch]
This is clear for $1$-node paths, and for $2$-node paths $\pi = (u, v)$ it follows by considering a witness graph $G$, and adding the edge $(u, v)$ with weight $\dist_G(u, v) - \eps$.
If $\eps > 0$ is selected sufficiently small -- say, less than the minimum positive difference between the lengths of any two simple paths in $G$ -- then this new edge will not change the shortest paths of $G$, besides those paths $\pi$ that have $u < v \in \pi$.
\end{proof}

\paragraph{Semisimplicity.}

A path $\pi$ that is neither a simple path nor a simple cycle can be split into two shorter paths, over its repeated node, as in Figure \ref{fig:semisimple}.
If $\pi$ has a weight $w(\pi)$, then we assign the same weight to each of the two new paths.
The two new paths (added together) will then have the same boundary as the original path.\footnote{Recall: since the cyclic part has the same start and end node $v$, this contributes a term of the form $-(v, v)$, which is ignored by the boundary operator.}

\begin{figure}[h]
\begin{center}
\begin{tikzpicture}

\draw [ultra thick, ->] (0, -0.4) -- (0, 2.4);
\draw [ultra thick] (-0.5, 1) circle [radius=0.5];
\draw [fill=black] (0, 1) circle [radius=0.15cm];

\draw [ultra thick, ->] (0.5, 1) -- (1.5, 1);

\begin{scope}[shift={(3, 0)}]
\draw [ultra thick, ->, gold] (0, -0.4) -- (0, 2.4);
\draw [ultra thick, teal] (-0.5, 1) circle [radius=0.5];
\draw [fill=black] (0, 1) circle [radius=0.15cm];
\end{scope}
\end{tikzpicture}
\end{center}
\caption{\label{fig:semisimple} Splitting a path into two paths over its repeated node does not change its boundary.}
\end{figure}

By repeating this operation, we can reach a path system in which every path is either a simple path or a simple cycle.
We will say that such a path system is \emph{semisimple} (emphasizing that simple cycles are not special cases of simple paths).
We note that a path system with a simple cycle is not strongly metrizable, so this operation does preserve strong metrizability, although only in the vacuous sense.

\paragraph{Skip-Free.}

Let us say that a path $\pi$ \emph{skips} a path $\pi'$ if: $\pi'$ has endpoints $(s, t)$ and at least one intermediate node, and also we have $(s, t) \in \pi$.
When this occurs, we can merge $\pi$ and $\pi'$ into a single path by inserting $\pi'$ as the $s \leadsto t$ subpath of $\pi$.
This operation again does not change the boundary of a path system.

\begin{figure}[h]
\begin{center}
\begin{tikzpicture}

\draw [fill=black] (0, 0.75) circle [radius=0.15cm];
\draw [fill=black] (0, 1.25) circle [radius=0.15cm];
\draw [ultra thick, ->, gold] (0, -0.4) -- (0, 2.4);
\draw [ultra thick, ->, teal] plot [smooth] coordinates {(0, 0.75) (-0.5, 0.5) (-1, 0.75) (-1, 1.25) (-0.5, 1.5) (0, 1.25)};

\draw [ultra thick, ->] (0.5, 1) -- (1.5, 1);

\begin{scope}[shift={(3, 0)}]
\draw [ultra thick, ->] plot [smooth] coordinates {(0, -0.4) (0, 0.75) (-0.5, 0.5) (-1, 0.75) (-1, 1.25) (-0.5, 1.5) (0, 1.25) (0, 2.4)};
\end{scope}

\end{tikzpicture}
\end{center}
\caption{Merging paths where one ``skips'' an edge of the other does not change its boundary.}
\end{figure}

This operation may also be applied to a weighted path system $S = (V, \Pi, w)$, in which case it only merges $\min\{w(\pi), w(\pi')\}$ units of path weight.
For example: if $w(\pi) < w(\pi')$, then the operation is to delete $\pi$, change the weight of $\pi'$ to $w(\pi') - w(\pi)$, and add the merged path with weight $w(\pi)$.

In a simple system, a skip implies that the system is inconsistent, which in turn implies that it is not strongly metrizable.
So, like the previous operation, this one vacuously preserves strong metrizability.
To summarize, we have:

\begin{fact} \label{fct:tosm}
For every path system $S = (V, \Pi, w)$, there is a reduced path system $S'$ with $\partial(S) = \partial(S')$.
Moreover, if $S$ is unweighted, simple, and consistent, then $S$ is strongly metrizable if and only if $S'$ is strongly metrizable.
\end{fact}

The only part of this fact that does not follow from the previous discussion is that, when we perform the previous operations until no longer possible, we will eventually terminate.
For completeness, this can be argued as follows:
\begin{itemize}
\item Each time we delete an isolated node or a trivial path, we reduce the number of nodes/paths by $1$, which can only occur a finite number of times.

\item Each time we apply the \emph{simplicity} splitting operation, or the \emph{skip-free} merging operation, notice that the potential function
$$\sum \limits_{\pi \in \Pi, \pi \text{ is not a cycle}} w(\pi) \cdot |\pi|$$
must decrease (except in the case where we split a non-simple cycle into two cycles, where it remains the same, but this happens only finitely often).
However, potential must remain nonnegative.\footnote{If path weights are integers, this clearly implies termination in finitely many rounds.  If path weights are rational, we could rescale them to be integers, so we also get termination in finitely many rounds; that setting will suffice for all applications in this paper.  If path weights are real numbers, then a slightly more involved argument is technically needed, e.g., considering the limit of the path systems reached by these operations.  This is not needed for our applications, so we omit it.}
\end{itemize}

\section{Characterizations of Strong Metrizability \label{sec:refutation}}

We will now start to prove some theorems that characterize the structure of strong metrizability.

\subsection{A Flow-Based Characterization Theorem}

The following background on flows is standard, although the phrasing in our algebraic framework is perhaps unusual.
\begin{definition} [Flows]
An \emph{$s \leadsto t$ flow} is a vector $f \in \rrp^{\vpp}$ satisfying $\partial(f) = \lambda(-s + t)$ for some $\lambda \in \rrp$.
The parameter $\lambda$ is called the \emph{value} of the flow $f$, and will be written $\lambda(f)$ where convenient.\footnote{The value of an $s \leadsto s$ flow (``a circulation'') is not well defined, and will not be needed.}
\end{definition}

\begin{definition} [Multiflows]
A $k$-multiflow $F = (f^1, \dots, f^k)$ over a graph $G = (V, E)$ is an ordered list of $k$ flows (with possibly different endpoints), where the support of each flow is contained in $E$.
\end{definition}

The following definition is new to this paper, and a centerpiece of our characterization:
\begin{definition} [Rigid Multiflows] \label{def:rigid}
A $k$-multiflow $F = (f^1, \dots, f^k)$ is \emph{rigid} if there is no distinct $k$-multiflow $F' = (f'^1, \dots, f'^k) \ne F$ satisfying:
\begin{enumerate}
\item $\partial\left(f^i\right) = \partial\left(f'^i\right)$ for all $1 \le i \le k$,\footnote{Recall that this is the \emph{first} boundary operator, so this is equivalent to saying that the flows have the same value and the same endpoints.}
\item $\sum \limits_{i=1}^{k} f^i = \sum \limits_{i=1}^k f'^i$ \end{enumerate}
\end{definition}

Our next goal is to prove the following characterization theorem.
It will occasionally be convenient in the following arguments to assume that the path systems in question contain at least two paths; the one-path case is essentially trivial, so it is ignored.

\begin{definition} [$\eff(S)$]
Let $S = (V, \Pi = \{\pi_1, \dots, \pi_k\}, w)$ be a path system and let $s_i \leadsto t_i$ be the endpoints of each $\pi_i$.
We define $\eff(S) = (f^1, \dots, f^k)$ to be the $k$-multiflow where for each $1 \le i \le k$, we have
$$f^i := w(\pi_i) \left(\partial\left(\pi_i\right) + \left(s_i, t_i\right)\right);$$
that is, $f^i$ is the flow along $\pi_i$ of value $w(\pi_i)$.\footnote{Technically,  $\eff(S)$ is an \emph{ordered list} whereas $\Pi$ is an \emph{unordered set}.  So part of the definition of $\eff(S)$ involves fixing an arbitrary order on the paths $\pi \in \Pi$.}
\end{definition}

\begin{theorem} \label{thm:rigidflow}
A simple path system $S = (V, \Pi)$ with $|\Pi| \ge 2$ is strongly metrizable if and only if $\eff(S)$ is rigid.
\end{theorem}

Let us run through some high-level intuition behind Theorem \ref{thm:rigidflow} before launching into the technical details.
Given a path system $S = (V, \Pi)$, imagine that we have some graph $G$ that contains $S$, and we place a walker at the start node of each $\pi \in \Pi$, and we want them to each walk to the corresponding end node.
We suggest to each walker that they take the route indicated by their path $\pi$.
However, the walkers discuss the plan and realize that they can collectively disobey our suggestions, rerouting their paths in some clever way, such that they still end at the same destinations and they collectively cover the same ground (i.e. the same number of walkers pass over each edge of $G$ as if the suggested paths $\pi \in \Pi$ had been taken).
If this rerouting is possible, then it acts as a certificate that not all $\pi \in \Pi$ can be a unique shortest path in any graph $G$ (regardless of its weights): some walkers chose alternate paths, and yet the total distance travelled by all walkers did not increase.

We can think of these alternate paths chosen by the walkers, which prove that $S$ is not strongly metrizable, as an unsplit multiflow $F'$ witnessing non-rigidity of $\eff(S)$.
In general, one can extend this argument to observe that \emph{fractional} reroutings also suffice as witnesses that some path system $S$ is not strongly metrizable, and these correspond to general (possibly split) multiflows $F'$ witnessing non-rigidity of $\eff(S)$.
Finally, Theorem \ref{thm:rigidflow} is saying that these fractional rerouting certificates of non-strong metrizability are exhaustive.
The intuition behind this last step is much like the intuition behind LP duality, although there is a technical difference: we rely on \emph{Gordan's Lemma}, a theorem of the alternative similar to Farkas' Lemma (the usual workhorse behind LP duality) that more precisely suits our needs.

\begin{lemma} [Gordan's Lemma] \label{lem:gordan}
Let $A \in \rr^{m \times n}$.  Then exactly one of the following two statements is true:
\begin{enumerate}
\item There exists $w \in \rr^n$ with $Aw > 0$, or
\item There exists $y \ne 0 \in \rrp^m$ with $yA = 0$.
\end{enumerate}
\end{lemma}

We now begin to prove Theorem \ref{thm:rigidflow}.
Let $S = (V, \Pi)$ be a path system hosted by a graph $G = (V, E)$, and let $s_i \leadsto t_i$ be the endpoints of each $\pi_i \in \Pi$.
We shall convert strong metrizability of $S$ into a statement about matrix math, so that we can apply Gordan's Lemma.
Here, the following matrices will be useful:
\begin{itemize}
\item Let $Q$ be the set of pairs of paths of the form $(q, \pi_i)$, where $\pi_i \in \Pi$, and $q \ne \pi_i$ is any distinct simple path between the same endpoints.
Let $A_Q \in \rr^{Q \times \vpp}$ be the linear map defined over its basis vectors $(q, \pi_i)$ by
$$\left(q, \pi_i\right) \cdot A_Q = \partial\left(q - \pi_i\right) \qquad \text{for all } \left(q, \pi_i\right) \in Q.$$
\item Let $\chi$ be the set of simple cycles over $V$.
Let $A_{\chi} \in \rr^{\chi \times \vpp}$ be the linear map defined by
$$c \cdot A_{\chi} = \partial(c) \qquad \text{for all } c \in \chi.$$

\item Define $A \in \rr^{(Q \cup \chi) \times \vpp}$ by appending $A_Q$ and $A_{\chi}$.
\end{itemize}

\begin{lemma} \label{lem:asm}
$S$ is strongly metrizable if and only if there exists $w \in \rr^{\vpp}$ with $Aw > 0$.
\end{lemma}
\begin{proof}
The path system $S$ is strongly metrizable if and only if we can equip the complete directed graph $G = (V, \vpp)$, with a weight vector $w \in \rr^{\vpp}$ such that:
\begin{enumerate}
\item For each $\pi_i \in \Pi$ and each simple path $q \ne \pi_i$ with the same endpoints, we have $w(\pi_i) < w(q)$, and
\item Every cycle $c$ in $G$ has $w(c) > 0$.
\end{enumerate}
By Fact \ref{fct:lenbd}, we can rewrite (1) as
\begin{align*}
\left\langle \partial\left(\pi_i\right) + (s, t), w \right\rangle &< \left\langle\partial(q) + (s, t), w \right\rangle && \text{for all } \left(q, \pi_i\right) \in Q\\
0 &< \left\langle \partial\left(q - \pi_i\right), w\right\rangle && \text{for all } \left(q, \pi_i\right) \in Q\\
0 &< A_Q w.
\end{align*}
Similarly, we can rewrite (2) as
\begin{align*}
\left\langle \partial(c), w\right\rangle &> 0  &&\text{for all } c \in \chi\\
A_{\chi} w &> 0.
\end{align*}
Thus $S$ is strongly metrizable if and only if there is $w \in \rr^{\vpp}$ with $A_Q w > 0$ and $A_{\chi} w > 0$, which is equivalent to $Aw > 0$.
\end{proof}

Our next goal is to equate rigidity to the other condition of Gordan's Lemma.
The following background fact about flow will be useful.

\begin{fact} \label{fct:pathdecomp}
For any $s \leadsto t$ flow $f$, there exists a positive linear combination of simple paths and simple cycles $p$, such that $\partial(p) = f - \lambda(f) \cdot (s, t)$.
\end{fact}

This fact can be proved, e.g., by a straightforward inductive process that peels off a simple path or simple cycle from the remaining flow in each round.
(This will not yield a \emph{unique} $p$ satisfying the fact -- the order of the peeling matters -- but any such $p$ may be used.)

\begin{lemma} \label{lem:arigidfor}
If $\eff(S)$ is not rigid, then there exists $y \ne 0 \in \rrp^{Q \cup \chi}$ with $yA = 0$.
\end{lemma}
\begin{proof}
Let $F'$ witness non-rigidity of $\eff(S)$, where
$$\eff(S) = (f^1, \dots, f^{k}) \ne F' = (f'^1, \dots, f'^{k}).$$
Let $(s_i, t_i)$ be the endpoints of each flow $f^i$ (which are the same as the endpoints of $f'^i$).
For each $1 \le i \le k$, let $p'^i$ be a path decomposition of $f'^i \in F'$ as in Fact \ref{fct:pathdecomp}; that is,
$$\partial\left(p'^i\right) = f'^i - \lambda(f'^i) \cdot (s_i, t_i).$$
Then lift $p'^i$ to a new vector $p'^i_Q \in \rrp^Q$, where
$p'^i_Q(q, \pi_i) := p'^i(q)$,
and define
$$y^Q := \sum \limits_{i=1}^k p'^i_Q \in \rrp^Q.$$
Note that the vectors $p'^i$ may also have simple cycles in $\chi$ in their support, which do not affect $y^Q$.
We will collect these separately: let $p'^i_{\chi}$ denote the restriction of $p'^i$ to its entries in $\chi$, and define
$$y^{\chi} := \sum \limits_{i=1}^k p'^i_{\chi} \in \rrp^{\chi}.$$
We then define $y \in \rrp^{Q \cup \chi}$ by concatenating $y^Q$ and $y^{\chi}$.
Note that $y \ne 0$, since $F' \ne F$, so it contains a flow that is not simply the path $\pi_i$, which contributes to $y$.
Additionally, we have:
\begin{align*}
yA &= y^Q A_Q + y^{\chi} A_{\chi}\\
&= \left( \sum \limits_{i=1}^k p'^i_Q \right) A_Q + \left( \sum \limits_{i=1}^k p'^i_{\chi} \right) A_{\chi}\\
&= \left( \sum \limits_{i=1}^k p'^i_Q \right) A_Q + \partial\left( \sum \limits_{i=1}^k p'^i_{\chi} \right) \tag*{by definition of $A_{\chi}$}\\
&= \left( \sum \limits_{i=1}^k \sum \limits_{(q, \pi_i) \in Q} p'^i(q) \cdot \left(q, \pi_i\right) \right) A_Q + \partial\left( \sum \limits_{i=1}^k p'^i_{\chi} \right)\\
&= \left( \sum \limits_{i=1}^k \sum \limits_{(q, \pi_i) \in Q} p'^i(q) \cdot \partial\left(q - \pi_i\right) \right) + \partial\left( \sum \limits_{i=1}^k p'^i_{\chi} \right) \tag*{by definition of $A_Q$}\\
&= \left( \sum \limits_{i=1}^k \sum \limits_{{\substack{q \text{ is a simple }\\s_i \leadsto t_i \text{ path}}}} p'^i(q) \cdot \partial(q) \right) + \partial\left( \sum \limits_{i=1}^k p'^i_{\chi} \right) - \left( \sum \limits_{i=1}^k \sum \limits_{\substack{q \text{ is a simple }\\s_i \leadsto t_i \text{ path}}} p'^i(q) \cdot \partial\left(\pi_i\right) \right)\\
&= \partial \left( \sum \limits_{i=1}^k p'^i \right) - \left( \sum \limits_{i=1}^k \sum \limits_{\substack{q \text{ is a simple }\\s_i \leadsto t_i \text{ path}}} p'^i(q) \cdot \partial\left(\pi_i\right) \right)\\
&= \partial \left( \sum \limits_{i=1}^k p'^i \right) - \left( \sum \limits_{i=1}^k \partial\left(\pi_i\right) \cdot \sum \limits_{\substack{q \text{ is a simple }\\s_i \leadsto t_i \text{ path}}} p'^i(q) \right)\\
&= \partial \left( \sum \limits_{i=1}^k p'^i \right) - \left( \sum \limits_{i=1}^k \partial\left(\pi_i\right) \cdot \lambda(f'^i) \right)\\
&= \partial \left( \sum \limits_{i=1}^k p'^i \right) - \left( \sum \limits_{i=1}^k \partial\left(\pi_i\right) \right) \tag*{since $\lambda(f'^i) = \lambda(f^i) = 1$}\\
&= \partial \left( \sum \limits_{i=1}^k p'^i \right) -  \left( \sum \limits_{i=1}^k f^i - \left(s_i, t_i\right) \right) \tag*{by definition of $\eff(S)$}\\
&= \partial \left( \sum \limits_{i=1}^k p'^i \right) + \sum \limits_{i=1}^k \left(s_i, t_i\right) - \sum \limits_{i=1}^k f'^i \tag*{by definition of (non-)rigidity}\\
&= 0,
\end{align*}
where the last equality follows from the equation in Fact \ref{fct:pathdecomp}.
\end{proof}

We then show the converse of the previous lemma:

\begin{lemma} \label{lem:arigidbac}
If there exists $y \ne 0 \in \rrp^{Q \cup \chi}$ with $yA = 0$, then $\eff(S)$ is not rigid.
\end{lemma}
\begin{proof}
If $S$ contains any non-simple paths $\pi$, then $\eff(S)$ is not rigid, and we are done.
This is because we could identify a contiguous simple cycle $x \subseteq \pi$, and then adjust the multiflow $\eff(S)$ by subtracting $\partial(x)$ from the flow corresponding to $\pi$, and adding $\partial(x)$ to any other flow in $\eff(S)$, yielding a distinct multiflow $F'$ that witnesses non-rigidity.
So we may assume in the following that all paths in $S$ are simple.

It will be convenient to assume without loss of generality that $y$ is scaled down such that
$$\sum \limits_{(q, \pi_i) \in Q} \left\langle y, (q, \pi_i)\right\rangle \le 1 \qquad \text{for any } i.$$
Let $\eff(S) = (f^1, \dots, f^k)$, and define a multiflow $F' := (f'^1, \dots, f'^k)$ where
$$f'^i := \left( \sum \limits_{(q, \pi_i) \in Q} \left\langle y, (q, \pi_i)\right\rangle \cdot  \left(\partial(q) + \left(s_i, t_i\right)\right) \right) + \left(1 - \sum \limits_{(q, \pi_i) \in Q} \left\langle y, (q, \pi_i)\right\rangle \right) \left(\partial\left(\pi_i\right) + \left(s_i, t_i\right)\right)+ \mathds{1}_{i=1} \cdot \partial\left( y^{\chi} \right);$$
that is, the flow $f'^i$ is a convex combination of the flow along $\pi_i$ and the flow along each path $q$ between the same endpoints.
The last term indicates: letting $y^{\chi}$ denote the restriction of $y$ to its indices in $\chi$, we add $\partial(y^{\chi})$ to $f'^1$, so that the cycle parts will be represented somewhere in the sum of flows.

We now claim that $F'$ witnesses non-rigidity of $\eff(S)$.
The first step is to observe that $F' \ne \eff(S)$.
Since each $\pi_i$ is a simple path, the corresponding flow has a unique path decomposition, which means it is distinct from $f'^i$ so long as the quantity
$$\sum \limits_{(q, \pi_i) \in Q} \left\langle y, (q, \pi_i)\right\rangle$$
is positive.
This must be the case for at least one $i$: if not, then $y$ is supported only on $\chi$, and since $y \ne 0$ it has nontrivial support.
So $yA = \partial(y^{\chi}) \ne 0$.

The second step is to observe the first property of (non-)rigidity, that for each $i$, the first boundaries of $f^i$ and $f'^i$ coincide.
This is a straightforward calculation from the definition of $f'^i$: we have
\begin{align*}
\partial(f'^i) &= \partial \left( \left( \sum \limits_{(q, \pi_i) \in Q} \left\langle y, (q, \pi_i)\right\rangle \left(\partial(q) + \left(s_i, t_i\right) \right) \right) + \left(1 - \sum \limits_{(q, \pi_i) \in Q} \left\langle y, (q, \pi_i)\right\rangle \right) \left(\partial\left(\pi_i\right) + \left(s_i, t_i\right)\right) + \mathds{1}_{i=1} \cdot \partial\left( y^{\chi} \right) \right)\\
&= \sum \limits_{(q, \pi_i) \in Q} \left\langle y, (q, \pi_i)\right\rangle \partial\left(s_i, t_i\right) + \left(1 - \sum \limits_{(q, \pi_i) \in Q} \left\langle y, (q, \pi_i)\right\rangle \right) \partial\left(s_i, t_i\right)\tag*{by Fact \ref{fct:bdbd}}\\
&= \partial\left(s_i, t_i\right)\\
&= \partial\left(f^i\right).
\end{align*}

The last step is to show the second property of non-rigidity, that the sum of flows coincides.
We have:
\begin{align*}
\sum \limits_{i=1}^{k} f'^i &= \sum \limits_{i=1}^{k} \left( \left( \sum \limits_{(q, \pi_i) \in Q} \left\langle y, (q, \pi_i)\right\rangle \left(\partial(q) + \left(s_i, t_i\right)\right)\right) + \left(1 - \sum \limits_{(q, \pi_i) \in Q} \left\langle y, (q, \pi_i)\right\rangle \right) \left(\partial\left(\pi_i\right) + \left(s_i, t_i\right)\right)  + \mathds{1}_{i=1} \cdot \partial\left( y^{\chi} \right)\right)\\
&= \sum \limits_{i=1}^{k} \left( \left( \sum \limits_{(q, \pi_i) \in Q} \left\langle y, (q, \pi_i)\right\rangle \left(\partial(q) - \partial(\pi_i)\right) \right) + \left(\partial\left(\pi_i\right) + \left(s_i, t_i\right)\right)  \right) + \partial\left( y^{\chi} \right)\\
&= \left(\sum \limits_{i=1}^{k} \sum \limits_{(q, \pi_i) \in Q} \left\langle y, (q, \pi_i)\right\rangle \left(\partial(q) - \partial\left(\pi_i\right)\right)\right) + \left(\sum \limits_{i=1}^{k}\partial\left(\pi_i\right) + \left(s_i, t_i\right)\right) + \partial\left( y^{\chi} \right)\\
&= \left(\partial\left( y^{\chi} \right) + \sum \limits_{(q, \pi) \in Q} \left\langle y, (q, \pi)\right\rangle \left(\partial(q) - \partial(\pi)\right)\right) + \left(\sum \limits_{i=1}^{k}\partial\left(\pi_i\right) + \left(s_i, t_i\right)\right) \\
&= \left(yA^{\chi} + yA^Q\right) + \left(\sum \limits_{i=1}^{k}\partial\left(\pi_i\right) + \left(s_i, t_i\right)\right) \tag*{by definition of $A$}\\
&= yA + \left(\sum \limits_{i=1}^{k} f^i \right)\\
&= \sum \limits_{i=1}^{k} f^i \tag*{since $yA=0$. \qedhere}
\end{align*}
\end{proof}

We now put it together:
\begin{proof} [Proof of Theorem \ref{thm:rigidflow}]
By Gordan's Lemma (\ref{lem:gordan}), there exists $w \in \rr^{\vpp}$ with $Aw > 0$ if and only if there does not exist $y \ne 0 \in \rrp^{Q \cup \chi}$ with $yA = 0$.
By Lemma \ref{lem:asm}, the former condition is equivalent to the statement that $S$ is strongly metrizable.
By Lemmas \ref{lem:arigidfor} and \ref{lem:arigidbac}, the latter condition is equivalent to the statement that $\eff(S)$ is rigid.
\end{proof}

Let us comment on one possible strengthening of Theorem \ref{thm:rigidflow}.
We have allowed the flow $F'$ witnessing non-rigidity to use arbitrary real values.
However -- because $\eff(S)$ has \emph{integer} values -- such a flow $F'$ exists iff there is one that specifically has \emph{rational} entries.
The freedom to assume this will be useful a bit later.

%
%
Finally, we note that Theorem \ref{thm:rigidflow} implies the following algorithmic corollary:
\begin{corollary} \label{cor:algorithm}
Given a path system $S = (V, \Pi)$, one can determine whether or not $S$ is strongly metrizable by solving a linear program on $m \cdot |\Pi|$ variables and $n|\Pi| + m|\Pi| + m$ constraints, where $n := |V|$, and $m$ is the number of ordered node pairs $(u, v)$ that appear consecutively on any path $\pi \in \Pi$.
\end{corollary}
\begin{proof}
We can trivially check that $S$ is simple, and reject if not.
After that check, strong metrizability is trivial if $|\Pi|=1$ (we check that the path is simple), so assume that $|\Pi| \ge 2$.
Let $G$ be the complete directed graph over vertices $V$, and let the ``capacity'' $c_{(u, v)}$ of each edge $(u, v)$ be a nonnegative integer indicating the number of times $(u, v)$ appears in any path $\pi \in \Pi$ (i.e. we count $\pi \in \Pi$ multiple times if $(u, v)$ appears in multiple places in $\pi$).
Let $P$ be the multiset of node pairs in $V$ that are the endpoints of a path $\pi \in \Pi$.

Now consider the following \emph{multicommodity flow feasibility} problem: the goal is to simultaneously push one unit of flow between each pair of nodes in $P$, obeying the constraints that each edge $(u, v)$ has exactly $c_{(u, v)}$ total flow passing through it.
We can set this up as a set of linear constraints in the standard way:
\begin{itemize}
\item We have a set of $m|\Pi|$ variables, each indicating the amount of flow for each path $\pi \in \Pi$ that is pushed along each edge of positive capacity, and
\item $n|\Pi| + m|\Pi| + m$ linear constraints, indicating that (1) incoming flow for each path is equal to outgoing flow for that path at each node (except for the source and sink of the path, at which these values differ by $1$), (2) the flow of each path on each edge is nonnegative, and (3) for each edge $(u, v)$ the total flow on that edge is exactly $c_{(u, v)}$.
\end{itemize}
These constraints are trivially feasible by the flow that pushes $1$ unit of flow corresponding to each path $\pi \in \Pi$ along the path $\pi$ itself.
By Theorem \ref{thm:rigidflow}, we have that $S$ is strongly metrizable if and only if this is the \emph{unique} way to satisfy these constraints.
It is shown in \cite{Appa02} that we can test uniqueness by solving an LP over the same constraints with a particular carefully-chosen objective function (which can be found in negligible runtime).
\end{proof}

\subsection{A Path System-Based Characterization Theorem}

Our next goal is to translate our characterization of strong metrizability back into the world of path systems, as follows:

\begin{theorem} \label{thm:rigidsystems}
A reduced path system $S = (V, \Pi)$ with $|\Pi| \ge 2$ is strongly metrizable if and only if it is simple, consistent, and there is no reduced weighted path system $S' = (V, \Pi', w') \ne S$ with $\partial(S') = \partial(S)$.
\end{theorem}

\begin{proof} [Proof of Theorem \ref{thm:rigidsystems}, Forwards Direction]
Let $S = (V, \Pi = \{\pi_1, \dots, \pi_k\})$ be a path system that is reduced and strongly metrizable, where each path $\pi_i$ has endpoints $s_i \leadsto t_i$.
Note that $S$ is simple and consistent.
The rest of the proof will proceed in contrapositive: suppose that there is a distinct reduced path system $S' = (V, \Pi', w') \ne S$ with $\partial(S') = \partial(S)$, and our goal will be to show that $S$ is not strongly metrizable.
We will do so by showing that the multiflow $\eff(S)$ is not rigid.
Let $\eff(S') =: \left(f'^1, \dots, f'^j\right)$
and define a multiflow $F'' = (f''^1, \dots, f''^k)$, where
$$f''^i := \sum \limits_{f' \in \eff(S') \text{ is an } s_i \leadsto t_i \text{ flow}} f'.$$
Additionally, for each cycle $x \in \Pi'$, we add $w'(x) \cdot \partial(x)$ to $f''^1$.
We will now prove that $F''$ witnesses non-rigidity of $\eff(S)$.\footnote{A technical detail here is that every path in $S'$ must either be a cycle or an $s_i \leadsto t_i$ path for some $i$, which thus contributes to one of the flows $f''^i$.  This is because, if there were an $s \leadsto t$ path $\pi' \in \Pi'$ where $(s, t)$ is \emph{not} one of the endpoint pairs $(s_i, t_i)$, then $\partial(S')(s, t)$ would be negative but $\partial(S)(s, t)$ would be nonnegative, contradicting that $\partial(S') = \partial(S)$.}

The first step is to observe that $F'' \ne \eff(S)$.
Since $S' \ne S$, there must be some index $i$ for which the set of $s_i \leadsto t_i$ paths in $\Pi'$ is not simply $\pi_i$, with weight $1$.\footnote{A technical detail here is that we cannot have the case where the set of $s_i \leadsto t_i$ paths in $S'$ is always just $\pi_i$ with weight $1$, but $S'$ is distinguished from $S$ by having some additional cycles.  This case is ruled out by the hypothesis $\partial(S) = \partial(S')$.}
Since $S$ is reduced, $\pi_i$ is simple.
Thus the sum of the corresponding flows in $\eff(S')$ is distinct from $f^i$.

The second step is to show that the first boundaries of the flows coincides.
We have:
\begin{align*}
\partial\left( f''^i \right) &= \partial \left( \sum \limits_{f' \in \eff(S') \text{ is an } s_i \leadsto t_i \text{ flow}} f' \right)\\
&= \partial \left( \sum \limits_{\pi' \in \Pi' \text{ is an } s_i \leadsto t_i \text{ path}} w'(\pi') \left( \partial\left(\pi'\right) + \left(s_i, t_i\right) \right) \right)\\
&= \sum \limits_{\pi' \in \Pi' \text{ is an } s_i \leadsto t_i \text{ path}} w'(\pi') \left( \partial\left(\partial(\pi')\right) + \partial\left(s_i, t_i\right) \right)\\
&= \sum \limits_{\pi' \in \Pi' \text{ is an } s_i \leadsto t_i \text{ path}} w'(\pi') \left( -s_i + t_i \right) \tag*{by Fact \ref{fct:bdbd}}\\
&= \left( \sum \limits_{\pi' \in \Pi' \text{ is an } s_i \leadsto t_i \text{ path}} w'(\pi') \right) \left( -s_i + t_i \right) \\
&= -\left\langle \partial(S'), (s_i, t_i)\right\rangle \cdot \left( -s_i + t_i \right) \tag*{since $S'$ is skip-free} \\
&= -\left\langle \partial(S), (s_i, t_i)\right\rangle \cdot \left( -s_i + t_i \right) \\
&= \left( -s_i + t_i \right) \tag*{since $S$ is skip-free}\\
&= \partial\left(f^i\right).
\end{align*}
Lastly, we need to show that the sum of flows coincide.
We have:
\begin{align*}
\sum \limits_{i=1}^k f^i &= \sum \limits_{i=1}^k \partial\left(\pi_i\right) + \left(s_i, t_i \right)\\
&= \partial\left(S\right) + \sum \limits_{i=1}^k \left(s_i, t_i \right)\\
&= \partial\left(S'\right) + \sum \limits_{i=1}^k \left(s_i, t_i \right)\\
&= \left( \sum \limits_{i=1}^j w'(\pi'_i) \partial(\pi'_i) \right) + \sum \limits_{i=1}^k \left(s_i, t_i \right) \tag*{by definition of $\partial(S')$}\\
&= \left( \sum \limits_{i=1}^j f'^i \right) - \left( \sum \limits_{i=1}^j w'(\pi'_i) \cdot (s'_i, t'_i) \right) + \sum \limits_{i=1}^k \left(s_i, t_i \right) \tag*{by definition of $f'^i$}\\
&= \left( \sum \limits_{i=1}^j f'^i \right) - \left( \sum \limits_{i=1}^k \left( \lambda(f''^i) - 1 \right) \cdot \left(s_i, t_i\right) \right) \tag*{by construction of $f''^i$}\\
&= \left( \sum \limits_{i=1}^j f'^i \right) - \left( \sum \limits_{i=1}^k \left( \lambda(f^i) - 1 \right) \cdot \left(s_i, t_i\right) \right)\\
&= \sum \limits_{i=1}^j f'^i \\
&= \sum \limits_{i=1}^k f''^i \tag*{by construction of $f''^i$. \qedhere}
\end{align*}
\end{proof}

\begin{proof} [Proof of Theorem \ref{thm:rigidsystems}, Backwards Direction]
Let $S = (V, \Pi = \{\pi_1, \dots, \pi_k\})$ be a path system that is reduced, simple, and consistent, but not strongly metrizable, and let $(s_i, t_i)$ be the endpoints of each path $\pi_i$.
By Theorem \ref{thm:rigidflow}, $\eff(S) = (f^1, \dots, f^k)$ is non-rigid, as witnessed by some multiflow
$$F' =: \left(f'^1, \dots, f'^k\right) \ne \eff(S).$$
Additionally, we may assume without loss of generality that each of these flows $f'^i \in F'$ is supported only on node pairs that appear adjacently along paths in $S$.
Since $S$ is reduced and contains an $s_i \leadsto t_i$ path for each $i$, this implies that $(s_i, t_i)$ is \emph{not} in the support of any of these flows.
Next, let $p'^i$ be a path decomposition for each $f'^i$, as in Fact \ref{fct:pathdecomp}, and let 
$$\overline{p} := \sum \limits_{i=1}^k p'_i.$$
Then, let
$$S' := \left(V, \supp\left(\overline{p}\right) =: \Pi', \overline{p}\right).$$
In other words, $S'$ is the path system whose paths are those in any of the path decompositions for $F'$, and whose weights are determined by the corresponding coefficients in $\overline{p}$.
Finally, apply the reduction steps to $S'$ so that it is reduced.

The next step is to show that $S \ne S'$.
Since $\eff(S) \ne F'$, there is an index $i$ for which the path decomposition $p'^i$ contains a path $\pi'_i \ne \pi_i$ in its support, where $\pi'_i$ is either a cycle or an $s_i \leadsto t_i$ path.
So we initially include $\pi'_i \in \Pi'$, and we need to show that a path of this form remains in $\Pi'$ even after we apply reduction steps.
There are four reduction steps to consider:
\begin{itemize}
\item Deleting isolated nodes does not affect $\pi'_i$.
\item As previously discussed, we may assume that there are no $1$- or $2$-node paths in $S'$.
Thus, no such paths are deleted.
\item Skip-free path merging cannot occur, because every path has endpoints $s_i \leadsto t_i$ for some $i$ (or it is a cycle), but the edge $(s_i, t_i)$ does not occur consecutively on any path in $S$ (since $S$ is reduced).

\item So, the only operations that change $\Pi'$ are semisimple path splitting.
This operation will leave a simple cycle in $\Pi'$, and since there are no simple cycles in $\Pi$, it preserves $S \ne S'$.
\end{itemize}
The last step is to show that $\partial(S) = \partial(S')$.
We have:
\begin{align*}
\partial(S) &= \sum \limits_{i=1}^k \partial \left( \pi_i \right) \\
&= \left( \sum \limits_{i=1}^k f^i \right) - \left( \sum \limits_{i=1}^k \left(s_i, t_i \right) \right)\\
&= \left( \sum \limits_{i=1}^k f'^i \right) - \left( \sum \limits_{i=1}^k \left(s_i, t_i \right) \right)\\
&= \left( \sum \limits_{i=1}^k \partial\left(p'_i\right) + \left(s_i, t_i\right) \right) - \left( \sum \limits_{i=1}^k \left(s_i, t_i \right) \right) \tag*{by Fact \ref{fct:pathdecomp}}\\
&= \sum \limits_{i=1}^k \partial\left(p'_i\right)\\
&= \partial\left(\overline{p}\right)\\
&= \partial\left(S'\right). \tag*{\qedhere}
\end{align*}
\end{proof}

We remark that Theorem \ref{thm:rigidsystems} also holds by the same proof in the case where $S$ is a \emph{multipath} system, i.e., it may include repeated paths in $\Pi$.
These repeated paths have no effect on strong metrizability, simplicity, or consistency.
However, repeated paths do affect the boundary, since $\partial(S)$ would count the contribution of repeated paths multiple times.

\section{Main Characterization Theorems \label{sec:simplification}}

Our goal in this section is to modify Theorem \ref{thm:rigidsystems} to find a mapping to $S, S'$ from a polyhedral pair $(T, T')$, as defined in the introduction.
We will first complete the proof in the directed setting, and then sketch how the entire argument adapts for the undirected setting.

\subsection{Forwards Direction of Directed Characterization}

Let $(T = (V, \Pi), T' = (V, \Pi'))$ be a polyhedral pair, derived from a two-colored polyhedron $Q$, and let $\phi : V \to U$ be a vertex map preserving distinctness ($\phi(T) \ne \phi(T')$).
In order to prove the first direction of Theorem \ref{thm:maindir}, our goal is to prove that $\phi(T)$ and $\phi(T')$ are not strongly metrizable.
By symmetry, it suffices to argue this for $\phi(T)$ only.
If $\phi(T)$ is not simple or consistent, then it is immediately not strongly metrizable, so we may assume in the following argument that it is.
We may also assume without loss of generality that $\phi$ is surjective, and so $\phi(T), \phi(T')$ do not contain isolated nodes.

First, we argue that $\partial(T) = \partial(T')$.
To see this, consider any pair of nodes $(u, v)$, and consider cases:
\begin{itemize}
\item Suppose that $(u, v)$ is used consecutively by a path in $\Pi \cup \Pi'$.
The structure of the two-colored polyhedron implies that each arc $(u, v)$ appears as an edge on exactly one colorful face and exactly one gray face of $Q$.
By the arc agreement property of polyhedral pairs, there is no path in $\Pi \cup \Pi'$ that has endpoints $(u, v)$.
Thus, $(u, v)$ is used consecutively by exactly one path in $\Pi$ and exactly one path in $\Pi'$, so we have
$$\langle \partial(T), (u, v) \rangle = \langle \partial(T'), (u, v) \rangle = 1.$$

\item Suppose instead that $(u, v)$ are the endpoints of a path in $\Pi \cup \Pi'$.
Arguing similarly, we will have
$$\langle \partial(T), (u, v) \rangle = \langle \partial(T'), (u, v) \rangle = -1.$$

\item If neither of the previous cases hold, then
$$\langle \partial(T), (u, v) \rangle = \langle \partial(T'), (u, v) \rangle = 0.$$
\end{itemize}

Now consider $\phi(T), \phi(T')$.\footnote{Technically, for the following argument, we need to treat these as \emph{multipath} systems: if two distinct paths are mapped together by $\phi$, they are still considered distinct.}
Since boundaries are preserved under vertex mapping, we will also have $\partial(\phi(T)) = \partial(\phi(T'))$.
Finally, we will argue that $\phi(T)$ is already reduced, and that it remains distinct from $\phi(T')$ even after $\phi(T')$ is reduced.
We consider the four reduction operations:
\begin{itemize}
\item \textbf{(Deleting an isolated node.)} We have assumed that there are no isolated nodes, so this operation will not occur.

\item \textbf{(Merging skips.)} There are no skips in $\phi(T)$, since we have assumed that it is simple and consistent.
A skip in $\phi(T')$ also may not occur, for the following reason.
A skip would occur if there is a pair of nodes $(u, v)$ that is used consecutively by a path in $\Pi'$, and another pair of nodes $(u', v')$ that is used as the endpoints of a path in $\Pi'$, and then $\phi(u) = \phi(u')$ and $\phi(v) = \phi(v')$.
However, there would then also be corresponding paths in $\Pi$ that use $(u, v)$ consecutively and $(u', v')$ as endpoints, since $\partial(T) = \partial(T')$.
So this skip would also occur in $\phi(T)$.
Thus, the skip-merging operation does not occur.

\item \textbf{(Splitting non-semisimple paths.)} We have assumed that $\phi(T)$ is simple, so this operation is not applied to $\phi(T)$.
This operation could be applied to $\phi(T')$, but after the operation $\phi(T')$ will contain a cycle.
So we would still have $\phi(T) \ne \phi(T')$, although it remains to ensure that the cycle is not then deleted in the following step.

\item \textbf{(Deleting trivial paths.)} All paths in $T, T'$ have at least three nodes, so this will remain true after the mapping.
Thus we will not delete trivial paths from $\phi(T)$.
We may, however, delete two-node paths from $\phi(T')$ in the case where they are created by a split of a non-semisimple path in the previous step.
If this occurs, observe that the number of pairs $(u, v)$ that appear adjacently on a path in $\phi(T')$ decreases by $1$ (counting repeats, and counting the case where $u=v$).
This quantity is preserved by splitting non-semisimple paths.
Thus, it will be smaller in $\phi(T')$ than in $\phi(T)$, implying distinctness.
\end{itemize}

To summarize: $\phi(T)$ is reduced and has at least two paths, and assuming that it is simple and consistent, there is a distinct reduced weighted system $\phi(T')$ with the same boundary.
Thus, by Theorem \ref{thm:rigidsystems}, $\phi(T)$ is not strongly metrizable.

\subsection{From Path Systems to Polyhedral Pairs \label{sec:polylift}}

We next prove the other direction of Theorem \ref{thm:maindir}.
We start with an arbitrary system $S = (V, \Pi)$ that is simple and consistent but not strongly metrizable, and our goal is to show that it contains the image of part of a polyhedral pair.
The following argument is standard in simplicial and cellular homology (see, e.g., \cite{LP12}, Proposition 2.1), although we will restate it from scratch since our formal objects are a bit different than usual.

\paragraph{Proof setup.}
First, we reduce $S$.
Since $S$ is simple and consistent, the only operations applied in the reduction are deleting isolated vertices and trivial paths.
The remaining system will be a subsystem of the original, and by Fact \ref{fct:pathdelete}, it will still not be strongly metrizable.

Since $S =: (V, \Pi)$ is now simple, consistent, reduced, and not strongly metrizable, by Theorem \ref{thm:rigidsystems} there is a distinct reduced weighted system $S' = (V, \Pi', w)$ with the same boundary.
Moreover, we may assume that the path weights $w$ are rational.
We then scale up the path weights in both systems $S, S'$ by a common factor, such that the weights in both systems are positive integers.
In the following, we will treat weights as \emph{multiplicities}: e.g., a path $\pi \in \Pi$ of weight $w$ will instead be treated as $w$ copies of $\pi$.
We may also cancel common paths from $S, S'$, without destroying either distinctness nor the equality $\partial(S) = \partial(S')$, so that the remaining systems have disjoint support.

We modify $S'$ in one more way before proceeding: if there are any simple cycles $c' \in \Pi'$, we attach $c'$ to any other path that shares a vertex with $c'$ (this is essentially the reverse of the semisimple modification).
Note that $S, S'$ remain distinct, since if this operation is performed at all, $S$ is reduced but $S'$ is not. 

\paragraph{Edge-matching.}

For each pair of nodes $(u, v)$ that appears consecutively on any path $\pi \in \Pi$, observe that (since $S$ is reduced and therefore skip-free) there are no $u \leadsto v$ paths, and therefore $\langle \partial(S), (u, v)\rangle$ is precisely the number of paths in $S$ (counting with multiplicity) that contain $(u, v)$ consecutively.
It follows that there are the same number of paths in $S'$ (counting with multiplicity) that use $(u, v)$ consecutively.
We may therefore fix an arbitrary matching $M[u, v]$ between the paths in $S$ and the paths in $S'$ that use $(u, v)$ consecutively.

Arguing similarly, for each pair of nodes $(u, v)$ that is the endpoints of any path $\pi \in \Pi$, the quantity $-\langle\partial(S), (u, v)\rangle$ is precisely the number of paths in both $S$ and $S'$ that have endpoints $(u, v)$.
We may again fix an arbitrary matching $M[u, v]$ between these paths.
For now, the matchings $M[u, v]$ are arbitrary.
However, as we will discuss in Section \ref{sec:cleanup}, there will be some aesthetic benefits to choosing the matchings in a particular way.

\paragraph{Construction of $T, T', \phi$.}

We are now ready to construct the polyhedral pair $T = (V^T, \Pi^T), T' = (V^T, \Pi'^T)$:

\begin{itemize}
\item Start with paths $\Pi^T, \Pi'^T$ corresponding to the paths of $\Pi, \Pi'$ respectively, but where each path is represented on a fresh set of vertices, so that the paths are currently pairwise node-disjoint.
Recall that paths in $\Pi'^T$ might be non-simple; if so, we still represent repeated vertices $v$ using a different vertex for each occurrence.
We will view these paths as cells, i.e., disjoint polygons in space.
Define $\phi : V^T \to V$ to map each vertex to its corresponding vertex in $V$.

\item In arbitrary order, consider the pairs of nodes $(u, v)$ that appear consecutively along paths in $\Pi, \Pi'$.
For each pair of paths $\pi \in \Pi, \pi' \in \Pi'$ that is matched in $M[u, v]$, glue together the occurrences of $(u, v)$ used by these respective paths.
In other words, we identify the respective nodes in $\pi, \pi'$ corresponding to $u \in V$, and identify the respective nodes in $\pi, \pi'$ corresponding to $v \in V$, and identify all points along the interval between these nodes.
Note that $\phi$ remains well-defined through this identification.

\item Similarly, consider the pairs of nodes $(u, v)$ that are the endpoints of paths in $\Pi, \Pi'$ in arbitrary order, and glue them over the matching $M[u, v]$ by the same process.
Again, note that $\phi$ remains well-defined.
\end{itemize}

This completes the construction.
The fact that
$\phi(T) = S \ne S' = \phi(T')$
is immediate from the setup, so it remains to confirm the remaining properties of polyhedral pairs promised by Theorem \ref{thm:maindir}.

\paragraph{Correctness of Construction.}

The arc agreement property, and the fact that paths in $T, T'$ occur in order (and properly respect orientation) around a face of $Q$, all follow directly from the edge-matching procedure in the construction.
Compactness (that the resulting manifold contains only finitely many cells, each of which is only finitely large) follows from finiteness of $\Pi^T, \Pi'^T$.
The main correctness property that needs proof is that the glued surface is topologically a 2-manifold without boundary.
See Figure \ref{fig:pinwheel} for a picture of the following argument on this point.

\boldmath
\begin{figure}[h]
\begin{center}
\begin{tikzpicture}

\draw [fill=gray!50] (2, 0) -- (0, 0) -- (1.414, 1.414);
\draw [fill=red!60] (1.414, 1.414) -- (0, 0) -- (0, 2);
\draw [fill=gray!50] (0, 2) -- (0, 0) -- (-1.414, 1.414);
\draw [fill=gold!60] (-1.414, 1.414) -- (0, 0) -- (-2, 0);
\draw [fill=gray!50] (-2, 0) -- (0, 0) -- (-1.414, -1.414);
\draw [fill=green!60] (-1.414, -1.414) -- (0, 0) -- (0, -2);
\draw [fill=gray!50] (0, -2) -- (0, 0) -- (1.414, -1.414);
\draw [fill=blue!60] (1.414, -1.414) -- (0, 0) -- (2, 0);

\draw [fill=black] (0, 0) circle [radius=0.3cm];

\draw [fill=black] (2, 0) circle [radius=0.15cm];
\draw [fill=black] (-2, 0) circle [radius=0.15cm];
\draw [fill=black] (0, 2) circle [radius=0.15cm];
\draw [fill=black] (0, -2) circle [radius=0.15cm];

\draw [fill=black] (1.414, 1.414) circle [radius=0.15cm];
\draw [fill=black] (-1.414, 1.414) circle [radius=0.15cm];
\draw [fill=black] (1.414, -1.414) circle [radius=0.15cm];
\draw [fill=black] (-1.414, -1.414) circle [radius=0.15cm];

\node at (0.924, 0.383) {$\pi'_0$};
\node at (0.383, 0.924) {$\pi_1$};
\node at (-0.383, 0.924) {$\pi'_1$};
\node at (-0.924, 0.383) {$\pi_2$};

\node at (-0.924, -0.383) {$\pi'_2$};
\node at (-0.383, -0.924) {$\pi_3$};
\node at (0.383, -0.924) {$\pi'_3$};
\node at (0.924, -0.383) {$\pi_0$};

\draw [ultra thick, >->] (1.697, 1.697) -- (0, 0) -- (2.4, 0);
\draw [ultra thick, >->] (1.697, 1.697) -- (0, 0) -- (0, 2.4);
\draw [ultra thick, >->] (-1.697, 1.697) -- (0, 0) -- (0, 2.4);
\draw [ultra thick, >->] (-1.697, 1.697) -- (0, 0) -- (-2.4, 0);
\draw [ultra thick, >->] (-1.697, -1.697) -- (0, 0) -- (-2.4, 0);
\draw [ultra thick, >->] (-1.697, -1.697) -- (0, 0) -- (0, -2.4);
\draw [ultra thick, >->] (1.697, -1.697) -- (0, 0) -- (0, -2.4);
\draw [ultra thick, >->] (1.697, -1.697) -- (0, 0) -- (2.4, 0);

\node [white] at (0, 0) {$v$};
\node [right] at (2, -0.3) {$x_0$};
\node [right] at (0.1, 2) {$x_1$};
\node [left] at (-2, -0.3) {$x_2$};
\node [left] at (-0.1, -2) {$x_3$};

\node [right=0.1cm] at (1.414, 1.414) {$u_0$};
\node [left=0.1cm] at (-1.414, 1.414) {$u_1$};
\node [left=0.1cm] at (-1.414, -1.414) {$u_2$};
\node [right=0.1cm] at (1.414, -1.414) {$u_3$};

\end{tikzpicture}
\end{center}

\caption{\label{fig:pinwheel} At each glued node $v$, the paths from $\Pi^T, \Pi'^T$ that use $v$ have been glued in a circular order as in this picture, implying flat topology at $v$.}
\end{figure}
\unboldmath

For each ordered pair of nodes $(u, v)$ that occurs adjacently clockwise around a cell corresponding to a path $\pi$ (either because $(u, v) \in \pi$ or because $\pi$ has endpoints $(v, u)$), we have glued it to exactly one pair on a cell $\pi'$.
Thus, the glued surface is locally flat at every point in the \emph{interior} of a cell or arc.
However, it still remains to show that the glued surface is also locally flat at its vertices, where several sides may be glued together.

Consider a vertex $v$ in the final glued surface, and let $\Pi_v, \Pi'_v$ be the subsets of paths from $\Pi^T, \Pi'^T$ (respectively) that contain $v$.
Choose a path $\pi_0 \in \Pi_v$ arbitrarily, and let $x_0$ be the node immediately following $v$ on $\pi_0$ (also considered before gluing; if $v$ is the last node on $\pi_0$ then we take $x_0$ to instead be the first node of $\pi_0$).
By construction, the arc $(v, x_0)$ is glued to an arc on some path $\pi'_0 \in \Pi'_v$.
This operation identifies the copies of $v$ and $x_0$ in $\pi_0, \pi'_0$.
Then, let $u_0$ be the node immediately preceding $v$ on $\pi'_0$ (or let $u_0$ be the last node on $\pi'_0$, if $v$ is the first).
Similarly, the arc $(u_0, v)$ was glued to one path in $\Pi_v$, and we may call the following node on that path $x_1$.
Since each arc $(v, x_i)$ or $(u_j, v)$ is \emph{uniquely} glued between two paths, we may repeat this process to uniquely generate a circular order of paths
$$\left(\pi_0, \pi'_0, \pi_1, \pi'_1, \dots, \pi'_{k-1}, \pi_k=\pi_0\right)$$
and vertices
$$\left(x_0, u_0, x_1, \dots, u_{k-1}, x_k=x_0\right),$$
halting once we repeat a path $\pi_0=\pi_k$.
Indeed, as this notation suggests, note that the \emph{first} path $\pi_0$ must be the one that is repeated.
This follows because each intermediate path $\pi_i$ has two edges incident to $v$, which have been glued to its adjacent paths $\pi'_{i-1}$ and $\pi'_i$, and hence were not also glued to $\pi'_{k-1}$.

This is a complete list of the paths whose copies of $v$ are identified in the final surface, and so all paths in $\Pi_v$ and $\Pi'_v$ must appear in this list.
The circular ordering of the paths implies that the surface has locally flat topology at $v$.

Finally, we remark that it is possible at this point that the construction has produced a \emph{disconnected} manifold.
In this case, one can discard all but one connected component without issue.
If desired, the formal correctness of this discarding follows from Section \ref{sec:cleanup} to follow (specifically Lemma \ref{lem:sss}).

\subsection{Undirected Setting}

For the most part, our results extend immediately from the directed setting to the undirected setting.
Indeed, our previous characterization theorems, such as Theorem \ref{thm:rigidflow}, extend immediately with respect to undirected rather than directed flows.
The only required change in the proof(s) is that we operate over the quotient space of $\rr^{\vpp}$ (or $\qq^{\vpp}$) in which we have $(u, v) = (v, u)$ for all pairs of distinct nodes $u, v$.

While the argument in the previous section still implies locally flat topology, the lack of direction on paths means that the choice of node $x_0$ following $v$ involves an arbitrary choice (i.e., we could choose $x_0$ to be either of the two nodes adjacent to $v$ on the path $\pi_0$).
The lack of path directions also means that the faces are not necessarily orientable, and so the resulting manifold may be non-orientable.

\subsection{Cleanup and Simplification \label{sec:cleanup}}

We will now demonstrate some additional properties that can be safely assumed for the polyhedral pairs in Theorems \ref{thm:maindir} and \ref{thm:mainintround}, if desired.
The following property of \emph{shadow-safety} is a technical but ultimately natural property one might like the vertex map $\phi$ to satisfy; we will see in a moment that it can be forced, and then that it allows us to simplify our polyhedral pairs, to a point.

\begin{definition}[Shadow-Safety]
Given path systems $S=(V_S, \Pi_S), T=(V_T, \Pi_T)$, a vertex map $\phi : V_T \to V_S$ is \emph{shadow-safe} if it has the following property.
For all paths $\pi \in \Pi_T$, and for all paths $q = (v_0, v_1, \dots, v_k)$ over $V_T$ with the same pair of endpoints as $\pi$ and with the property that every pair of adjacent vertices $(v_i, v_{i+1})$ appears adjacently in some path in $\Pi_T$, we have
$$\text{there exists } \pi' \in \Pi_S \text{ with } \phi(q) \subseteq \pi' \qquad \text{if and only if} \qquad \pi = q.$$
\end{definition}

In other words, shadow-safety says that $\phi$ maps each path $\pi \in \Pi_T$ to a subpath of a path from $\Pi_S$, and also it maps each path that competes with $\pi$ to a non-subpath.
We show next that it may be enforced, alongside some basic simplifications:


\begin{lemma} \label{lem:ssprops}
In the construction of $(T, T')$ in Section \ref{sec:polylift}, suppose that the choice of matchings $\{M[u, v]\}$ is selected in such a way that we maximize the final number of vertices $|V_T|$, among all valid choices of matchings.
Then:
\begin{itemize}
    \item there is no parallel edge; i.e., no two paths in $T$ use the same adjacent pair of nodes $(u, v)$,
    \item the map $\phi:V_T\to V_S$ is shadow-safe.
\end{itemize}
\end{lemma}
\begin{proof}
We will establish both properties by arguing that, if they are violated, then there is a way to swap choices of matchings for some class $M[u, v]$ in a way that increases the final number of nodes $|V_T|$.

\paragraph{Parallel edges.}
Suppose for contradiction that an ordered pair $(u,v)$ occurs consecutively on two distinct paths in $\Pi_T$.
Let $e$ and $f$ denote the two corresponding parallel $(u, v)$ edges in the polyhedron.
Since $\phi(e) = \phi(f)$, the edges $e$ and $f$ belong to the same matching class $M[\phi(u),\phi(v)]$, and thus we have the freedom to exchange the choice of gray edges matched to these two edges.
After this exchange, we will no longer identify the copies of $u$ and $v$ used by $e$ and $f$, and so the number of nodes $|V_T|$ increases, contradicting maximality.

\paragraph{Shadow-safety.}
Let $\pi=(x_0,x_1,\ldots,x_k) \in \Pi_T$.
By construction, its image $\phi(\pi)$ is exactly equal to some path $\pi_S \in \Pi_S$, thus witnessing the forward direction of shadow-safety.
For the reverse direction, let
\[
    q=(x_0=y_0,y_1,\ldots,y_\ell=x_k)\neq\pi
\]
be a competing path with the same endpoints as $\pi$, and suppose for contradiction that $\phi(q)$ is a subpath of some $\pi'_S\in\Pi_S$.
The paths $\pi_S$ and $\pi'_S$ each contain the images of the two common endpoints $(x_0=y_0, x_k=y_{\ell})$ in the same order.
Since $S$ is simple and consistent, the contiguous subpaths of $\pi_S$ and $\pi'_S$ between these nodes agree.
In particular, this implies that $\ell=k$; that is, these subpaths have the same number of nodes.

Since $q\neq\pi$, we may choose two indices $i, j$ with $i < j$ such that $x_i=y_i$ and $x_j=y_j$, but where the paths are internally node-disjoint on their subpaths between these nodes.
For each index $i \le r \le j-1$, let $e_r = (x_r, x_{r+1})$ be the colorful edge used by $\pi$, and let $f_r = (y_r, y_{r+1})$ be the distinct colorful edge used by $q$.
Note that for all $i \le r \le j-1$, the $r^{th}$ edges $(\phi(x_r)=\phi(y_r), \phi(x_{r+1})=\phi(y_{r+1}))$ belong to the same matching class, and so we may switch their assigned partners in the matching.
We will next argue that this structure provides an opportunity for us to perform matching switches to increase $|V_T|$, again contradicting maximality.

Start by switching the matched partners of $e_i$ and $f_i$, noting that these edges have a common start node $x_i=y_i$, but they may have distinct end nodes.
The switch will place their start nodes in different circular groups, thus causing them to split into two distinct nodes in $T$.
If the switch does \emph{not} also cause their distinct end nodes to merge into one, or in the case where their end nodes were already identical, then we have increased $|V_T|$ by $1$, reaching contradiction.
So we may assume that the switch causes their end nodes to merge.
We may then continue by switching $e_{i+1}$ and $f_{i+1}$, and so on.
Eventually, we will switch $e_{j-1}$ and $f_{j-1}$, which have common end nodes $x_j=y_j$, which increases $|V_T|$ and reaches contradiction.
\end{proof}

\begin{lemma}
Let $(T,T')$ and $\phi:V_T\to V_S$ be obtained from the construction in Section \ref{sec:polylift}.
By splitting nodes, we may obtain another polyhedral pair $(\widehat T,\widehat T')$ and vertex map $\widehat\phi:V_{\widehat T}\to V_S$ such that every face of
the resulting polyhedron is simple and contains at least three nodes.
Moreover, if initially $T, \phi$ satisfy the properties in Lemma \ref{lem:ssprops}, then so do $\widehat T, \widehat\phi$.
\end{lemma}
\begin{proof}
By construction, the paths $\pi \in \Pi_T$ have the property that $\phi(\pi)=\pi_S$ for some $\pi_S \in \Pi_S$.
Since $S$ is simple, this implies that $\pi_S$ is simple, and so $\pi$ must be simple as well.
Additionally, since $S$ is reduced, $\pi_S$ must have at least three distinct nodes, and so $\pi$ has at least three distinct nodes as well.
It thus suffices to focus on the paths $\pi' \in \Pi_{T'}$, which may not be simple.

Consider a non-simple path $\pi' \in \Pi'_T$ that repeats a node $v$.
Then $\pi'$ contains a contiguous subpath $C$ that is a simple cycle starting and ending at $v$.
Split $v$ into two new nodes $v_1, v_2$, with these consecutive instances of $v$ replaced by $v_1, v_2$ respectively.
The gray edges incident to $v$, which we may call $(v, u_1), (u_2, v) \in C$, are thus replaced by $(v_1, u_1), (u_2, v_2)$.
We also replace occurrences of $v$ in the matched colorful edges so that the matchings do not need to be re-selected.
In other words:
\begin{itemize}
\item The colorful edge matched to $(v, u_1)$ is also replaced by $(v_1, u_1)$, and the colorful edge matched to $(u_2, v)$ is also replaced by $(u_2, v_2)$,
\item For each other node $x$ with an edge $(x, v)$, we replace all gray and colorful edges of this form with the same copy, either $(x, v_1)$ or $(x, v_2)$.
The copy is selected in such a way that the manifold structure is preserved; for example, all edges that remain cyclically attached to $(v, u_1)$ after the split would use $v_1$.
\end{itemize}
Repeat this splitting process until $\Pi'_T$ is simple; the fact that each corresponding gray face has at least three nodes follows from the same argument as before.
Note that we will terminate in a finite number of steps, since each vertex split removes one instance of a repeated node in a path in $\Pi'_T$.

It is immediate that each split preserves the properties of polyhedral pairs, and that splitting cannot create parallel edges (as resolved by the previous lemma).
It remains to prove that the vertex map $\phi$ remains shadow-safe after each splitting operation.
The forwards direction of shadow-safety is unaffected by splitting, so it is again immediate from the fact that $\phi$ is initially shadow-safe.

For the reverse direction, let $\pi \in \Pi_T$ (after splitting) and let $q$ be a competing path for
$\pi$.
Since $\phi$ was previously shadow-safe, we have that $\phi(q)$ is not a subpath of any path in $\Pi_S$, except perhaps for the case where $\pi, q$ differ only on the node that was just split: one uses $v_1$, the other uses $v_2$, but the paths are otherwise identical.
That is, without loss of generality, we have
$$\pi=(\ldots,u,v_1,w,\ldots)  \qquad\text{and}\qquad q=(\ldots,u,v_2,w,\ldots).$$
To rule out this remaining case, note that we would have colorful edges $(u, v_1)$ and $(u, v_2)$, but by construction all colorful edges $(u, v)$ would be replaced by the same copy of $v$.
\end{proof}

Let us now discuss the favorable consequences of shadow-safety.
The following lemmas show that it interacts nicely with strong metrizability:

\begin{lemma} \label{lem:sss}
For path systems $T = (V_T, \Pi_T), S = (V_S, \Pi_S)$ where $S$ is simple and consistent, if a map $\phi : V_T \to V_S$ is shadow-safe, then for any subsystem $R = (V_R, \Pi_R) \subseteq T$, the restriction of $\phi$ to $V_R$ is also shadow-safe.
\end{lemma}
\begin{proof}
We will argue that the subpath operations all preserve shadow-safety:
\begin{itemize}
\item If we delete a path $\pi \in \Pi_T$, shadow-safety is preserved trivially.
\item If we delete a node $v \in V_T$, we will clearly still have the property that for all $\pi_T \in \Pi_T$, $\phi(\pi_T)$ is a subpath of some path in $\Pi_S$.
Moreover, for any competing path $q$, there will be a corresponding path $q'$ of the same length after the deletion, so $\phi$ will remain shadow-safe.

\item The more involved operation is when we split a path $\pi_T \in \Pi_T$ into two contiguous subpaths $\pi_1, \pi_2$.
Since $\phi(\pi_T)$ is a subpath of some path $\pi_S \in \Pi_S$, we immediately have that $\phi(\pi_1), \phi(\pi_2)$ are subpaths of $\pi_S$, as desired.
Additionally, let $q \ne \pi_1$ be a path that competes with $\pi_1$.
Observe that:
\begin{itemize}
\item The concatenation $q \circ \pi_2 \ne \pi$ competes with $\pi$, and hence it is not a subpath of $\pi_S$.
Since $\phi(\pi_2)$ is a subpath of $\pi_S$, it follows that $\phi(q)$ is not a subpath of the prefix of $\pi_S$ up to the endpoint of $\phi(q)$.

\item Additionally, we claim that $\phi(q)$ cannot be a subpath of any other path $\pi'_S \in \Pi_S$.
If it were, then $\pi_S, \pi'_S$ would violate the consistency of $S$.
\end{itemize}
This shows that no distinct path $q$ that competes with $\pi_1$ has the property that $\phi(q)$ is a subpath of any path in $\Pi_S$.
A similar argument applies to paths that compete with $\pi_2$. \qedhere
\end{itemize}
\end{proof}

\begin{lemma} \label{lem:ssinherit}
If $T = (V_T, \Pi_T)$ is not strongly metrizable, and there exist $S = (V_S, \Pi_S)$ and a shadow-safe map $\phi : V_T \to V_S$, then $S$ is also not strongly metrizable.
\end{lemma}
\begin{proof}
We show the contrapositive.
Suppose $S = (V_S, \Pi_S)$ is strongly metrizable, as witnessed by a graph $G = (V_S, E, w)$.
Define a graph $G' = (V_T, E', w')$ by including each edge $(u, v) \in E'$ that appears consecutively in any path in $\Pi_T$, and where $w(u, v) := \dist_G(\phi(u), \phi(v))$.
We will now argue that $G'$ witnesses strong metrizability of $T$.

Consider a path $\pi = (v_0, v_1, \dots, v_k) \in \Pi_T$, and any alternate path $q = (v_0=v'_0, v'_1, \dots, v'_j=v_k)$ through $G'$ with the same endpoints as $\pi$.
Our goal is to show that $q$ is longer than $\pi$ in $G'$.
Observe that:
\begin{itemize}
\item We have $$w'(\pi) = \sum \limits_{0 \le i \le k-1} \dist_G(\phi(v_i), \phi(v_{i+1})) = \dist_G(\phi(v_0), \phi(v_k)),$$
where the latter equality follows since $\phi(\pi)$ is a subpath of a path in $\Pi_S$, and hence the vertices of $\phi(\pi)$ lie in order along a shortest path in $G$.

\item We have $$w'(q) = \sum \limits_{0 \le i \le j-1} \dist_G(\phi(v'_i), \phi(v'_{i+1})) > \dist_G(\phi(v'_0), \phi(v'_j)) = \dist_G(\phi(v_0), \phi(v_k)),$$ where the strict inequality follows since (by shadow-safety of $\phi$) $\phi(q)$ is not a subpath of any path in $\Pi_S$, and hence the vertices of $\phi(q)$ do not lie in order along the shortest path between its endpoints $(\phi(v_0), \phi(v_k))$.
\end{itemize}
Thus $w'(q) > w'(\pi)$, completing the proof.
\end{proof}

As a consequence, we may enforce \emph{minimality} of our obstructions, to a point:

\begin{theorem} [Reframing of Main Characterization]
A path system $S = (V_S, \Pi_S)$ is strongly metrizable (in the directed or undirected settings) if and only if it is simple and consistent, and there is no polyhedral pair $(T = (V_T, \Pi_T), T')$ and shadow-safe map $\phi : V_T \to V_S$.
Moreover, we may assume that $T$ is minimal, in the sense that no proper subsystem $R \subsetneq T$ participates in a polyhedral pair $(R, R')$.
\end{theorem}
\begin{proof}
Following the proof from Section \ref{sec:polylift}, we construct $(T, T')$ and (by Lemma \ref{lem:ssprops}) we may also have that $\phi$ is shadow-safe.
By Lemma \ref{lem:sss}, if there exists a polyhedral pair $(R, R')$ with $R \subsetneq T$, then $\phi$ is also a shadow-safe map from $R$ to $S$.
Thus it suffices to use $R$ as an obstruction in place of $T$.
\end{proof}

\subsection{The Directed Acyclic Setting}

In order to study strong metrizability in DAGs, we can define acyclic path systems as follows:

\begin{definition} [Acyclic Path Systems]
A path system $S = (V, \Pi)$ is \emph{acyclic} if we can choose a total ordering of its nodes (called a \emph{topological ordering}) such that the nodes in each $\pi \in \Pi$ appear in this order.\footnote{An equivalent definition is that $S$ is acyclic iff it does not contain a ``directed cycle'' -- that is, a system consisting of $k$ nodes and $k$ paths with two nodes each, arranged in a directed cycle -- as a subsystem.}
\end{definition}

These are a special case of directed path systems, and hence our directed characterization of strong metrizability still applies.
However, we can simplify the characterization a bit:

\begin{corollary}
Let $S = (V, \Pi)$ be a directed acyclic strongly metrizable path system, and let $S' = (V, \Pi')$ be a directed path system where each $\pi' \in \Pi'$ is obtained by taking circular shifts of the node ordering of some corresponding path $\pi \in \Pi$.
Then $S'$ is strongly metrizable if and only if it is consistent.
\end{corollary}
\begin{proof}
Let $Q$ be a two-colored abstract polyhedron.
Let $(T, T')$ be path systems that respectively map to the colorful and gray faces of $Q$, as in the definition of polyhedral pairs, but without necessarily satisfying the last arc-agreement condition of polyhedral pairs.
We next argue that $(T, T')$ satisfy the last arc agreement property, forming a polyhedral pair, if and only if they contain the same number of paths as each other:
\begin{itemize}
\item Suppose that $T, T'$ contain the same number of paths as each other.
Then no face of $Q$ may have two distinct endpoint edges on its boundary, since this would imply that some other face of $Q$ has no endpoint edges on its boundary, implying a cycle.
Thus we cannot have a node pair $(u, v)$ that is used as both an endpoint and non-endpoint edge of $Q$, implying that $(T, T')$ satisfy arc agreement.

\item Conversely, suppose that $(T, T')$ form a polyhedral pair, satisfying the last arc agreement condition as well.
This implies that each face of $Q$ contains at most one endpoint edge on its boundary.
Moreover, since neither $T$ nor $T'$ may contain a cycle, each face of $Q$ contains at least one endpoint edge on its boundary.
So each face of $Q$ has exactly one endpoint edge on its boundary, implying that these edges induce a perfect matching between the gray and colorful faces of $Q$.
Thus $T$ and $T'$ have the same number of paths.
\end{itemize}

So, in the acyclic setting, we may replace the last condition in the definition of polyhedral pairs with ``$T$ and $T'$ have the same number of paths."
With this change, the definition is invariant to circular shifts in path orders.
Since $S$ is strongly metrizable, it therefore avoids images of polyhedral pairs as subsystems, and so $S'$ does as well.
So $S'$ is strongly metrizable if and only if it remains consistent.
\end{proof}

%% file: acks.tex
\section*{Acknowledgments}

I have quite a lot of people to thank for useful technical discussions on the ideas of this paper.
In no particular order, these people are: Virginia Vassilevska Williams, Michael Cohen, Josh Alman, Dylan McKay, Andrea Lincoln, Jerry Li, Ofer Grossman, Govind Ramnarayan, Aviv Adler, Sitan Chen, Atish Agarwala, and Noga Alon.
I am grateful to Hangyu Xu and Gary Hoppenworth for independently alerting me to a bug in a previous draft of the paper.
I am also grateful to Ryan Williams and an anonymous reviewer for writing advice that has improved the presentation of this work.